\documentclass{lmcs} 
\pdfoutput=1
\usepackage[utf8]{inputenc}

\usepackage{lastpage}
\lmcsdoi{19}{3}{10}
\lmcsheading{}{\pageref{LastPage}}{}{}%
{Mar.~15,~2022}{Aug.~03,~2023}{}


\usepackage{tikz}
\usetikzlibrary{shapes,snakes,positioning}
\usepackage{hyperref}
\theoremstyle{plain} 

\begin{document}
\newcommand{\act}{\mathsf{Act}}
\newcommand{\Min}{\mathrm{Min}}
\newcommand{\Max}{\mathrm{Max}}
\newcommand{\colour}{\mathsf{color}}
\newcommand{\play}{\mathrm{play}}
\newcommand{\source}{\mathsf{source}}
\newcommand{\target}{\mathsf{target}}
\newcommand{\col}{\mathrm{col}}
\newcommand{\lab}{\mathsf{lab}}
\newcommand{\dist}{\mathrm{Dist}}
\newcommand{\val}{\mathrm{val}}
\newcommand{\valbf}{\mathbf{val}}
\newcommand{\GDP}{\mathrm{GDP}}
\newcommand{\vphi}{\varphi}
\newcommand{\sh}{\mathbf{s}}
\newcommand{\x}{\mathbf{x}}
\newcommand{\y}{\mathbf{y}}
\newcommand{\g}{\mathbf{g}}
\newcommand{\diam}{\mathsf{diam}}
\newcommand{\Cons}{\mathsf{Cons}}
\newcommand{\Val}{\mathbf{Val}}

\title{Continuous Positional Payoffs\rsuper*}
\titlecomment{{\lsuper*}An extended abstract of this paper was presented at the CONCUR 2021 conference.}

\author[A.~Kozachinskiy]{Alexander Kozachinskiy}	
\address{Centro Nacional de Inteligencia Artificial, Chile \and
    Instituto de Ingenier\'ia Matem\'atica y Computacional, Universidad Cat\'olica de Chile \and
    Instituto Milenio Fundamentos de los Datos, Chile}	
\email{alexander.kozachinskyi@cenia.cl}  
\thanks{The author is funded by ANID - Millennium Science Initiative Program - Code ICN17002, and the National Center for Artificial Intelligence CENIA
FB210017, Basal ANID}	

\tikzset{
  gamegraph/.style={
    transition/.style ={
      ->,
      thick,
      shorten <= 1pt,
      shorten >= 1pt,
      overlay,
        every node/.append style={
          outer sep=5pt,
        },
    },
  },
}





\begin{abstract}
  \noindent What payoffs are positionally determined for deterministic two-player antagonistic games on finite directed graphs? In this paper we study this question for payoffs that are continuous.  The main reason why continuous positionally determined payoffs are interesting is that they include the multi-discounted payoffs.

We show that for continuous payoffs, positional determinacy is equivalent to a simple property called prefix-monotonicity. We provide three proofs of it, using  three major techniques of establishing positional determinacy -- inductive technique, fixed point technique and strategy improvement technique. A combination of these approaches provides us with better understanding of the structure of continuous positionally determined payoffs as well as with some algorithmic results.
\end{abstract}

\maketitle


\section{Introduction}
\label{sec:introduction}

We study two-player turn-based games on finite directed graphs. In these games, two players called Max and Min travel over nodes of a given graph along its edges for infinitely many turns. In each turn, one of the players decides where to go next, and which of the two depends on a predetermined partition of the nodes between the players.

The game is of infinite duration.
As a result of the game, we get an infinite path in our graph. Each infinite path is mapped to a real number called its \emph{reward}, according to some \emph{payoff function} (or, for brevity, a payoff). The larger the reward is the more Max is happy; on the contrary, Min wants to minimize the reward.

We consider only payoffs that are defined through \emph{edge labels}. Namely, we first fix some finite set $A$ of labels. Then we label edges of our game graph by elements of $A$. After this, any bounded function $\vphi\colon A^\omega\to\mathbb{R}$ can be viewed as a payoff in our graph. Namely, it takes an infinite path,  considers an infinite word over $A$ ``written''  on this path, and applies $\vphi$ to this infinite word.

Fix a \emph{strategy} of one of the players (that is, an instruction how to play in all possible developments of the game). If this is a strategy of Max, then its \emph{value} is the infimum of the payoff over all infinite paths that can occur in a play with this strategy. Similarly, if this a strategy of Min, then its value is the supremum of the payoff over all infinite paths that can occur in a play with this strategy.
Values are a standard way of measuring the ``worst-case'' quality of a strategy. 

A strategy of Max is called \emph{optimal} if its value is at least as large as the value of any other Max's strategy. Similarly, a strategy of Min is called \emph{optimal} if its value does not exceed the value of any other Min's strategy.

Observe that the value of any Min's strategy is at least as large as the value of any Max's strategy. A pair $(\sigma,\tau)$ of a Max's strategy $\sigma$ and a Min's strategy $\tau$ is called an \emph{equilibrium} if the value of $\sigma$ \emph{equals} the value of $\tau$. Both strategies appearing in an equilibrium must be optimal -- one proves the optimality of the other. In this paper, we only study so-called \emph{determined} payoffs. These are payoffs that have an equilibrium in all game graphs.

Games that we are studying proceed for infinitely many turns, so strategies in them might be rather complicated. An area of \emph{strategy complexity} classifies payoffs according to their ``simplicity''. A ``simple'' payoff always admits an equilibrium of two  ``simple'' strategies. This has been studied in various settings~\cite{bouyer2020games}. We study, perhaps,  the most well-established one, where by ``simple'' strategies we understand \emph{positional} strategies.

 A strategy of one of the players is called positional if for every node there exists a single out-going edge such that this strategy always uses this edge from this node (of course, we only require this for nodes from where the corresponding player is the one to move). Essentially, a positional strategy is a strategy with no memory --  at every location, it completely ignores the previous development of the game. Now, a payoff is called \emph{positionally determined} if every game graph has an equilibrium of two positional strategies w.r.t.~this payoff.

\medskip

A lot of works are devoted to concrete positionally determined payoffs that are of interest in other areas of computer science. Classical examples of such payoffs are parity payoffs, mean payoffs and (multi-)discounted payoffs \cite{ehrenfeucht1979positional,mostowski1991games,mcnaughton1993infinite,shapley1953stochastic}. Their applications range from logic, verification and finite automata theory~\cite{emerson1991tree,gradel2002automata} to decision-making~\cite{puterman2014markov,sutton1998introduction} and algorithm design~\cite{calude2020deciding}.

Along with this specialized research, in~\cite{gimbert2004can,gimbert2005games} Gimbert and Zielonka undertook a thorough study of positionally determined payoffs in general. In \cite{gimbert2004can} they showed that all so-called \emph{fairly mixing} payoffs are positionally determined. They also demonstrated that virtually all classical positionally determined payoffs are fairly mixing. Next, in~\cite{gimbert2005games} they established a property of payoffs which is \emph{equivalent} to positional determinacy. Unfortunately, this property is far more technical than the fairly mixing property, and it is hard to use it in practice. Still, this property has a remarkable feature: if a payoff does not satisfy it, then this payoff is not positionally determined in some \emph{one-player} game graph (i.e., in a game graph where one of the players owns all the nodes). As Gimbert and Zielonka indicate, this means that to establish the positional determinacy of a payoff, it is enough to do so only for one-player game graphs.

 Unfortunately, these results do not quite answer how positionally determined payoffs are arranged in general. The goal of the present paper is to make progress in this direction -- at least for payoffs that satisfy some natural additional properties. One such property studied in the literature is called \emph{prefix-independence}~\cite{colcombet2006positional,gimbert2007pure}. A payoff is prefix-independent if it is invariant under throwing away any finite prefix from the input word. For instance, the parity and the mean payoffs are prefix-independent.

In~\cite{gimbert2004can}, Gimbert and Zielonka briefly mention another interesting additional property, namely, \emph{continuity}. They observe that the multi-discounted payoffs are continuous (they utilize this in showing that the multi-discounted payoffs are fairly mixing). In this paper, we study continuous positionally determined payoffs in more detail.  A payoff is continuous if its range converges to just a single point as more and more initial letters of its input (which an infinite word over the set of labels) are getting fixed. This contrasts with prefix-independent payoffs (such as the parity and the mean payoffs), for which any initial finite segment is irrelevant. Thus, continuity serves as a natural property which separates the multi-discounted payoffs from other classical positionally determined payoffs. This is our main motivation to study continuous positionally determined payoffs in general, besides the general importance of the notion of continuity.

\bigskip

 We show that for continuous payoffs, positional determinacy is equivalent to a simple property which we call \emph{prefix-monotonicity}. A payoff $\vphi$ is prefix-monotone if there are no two infinite words $\alpha$ and $\beta$ and no two finite words $x$ and $y$ such that $\vphi(x\alpha) > \vphi(x\beta)$ and $\vphi(y\alpha) < \vphi(y\beta)$.

 A proof of the fact that any continuous positionally determined payoff is prefix-monotone can be found in Section \ref{sec:results}. We give three different proofs of the opposite direction of our main result, using three major techniques of establishing positional determinacy:
\begin{itemize}
\item \emph{An inductive argument.} Here we use a sufficient condition of Gimbert and Zielonka~\cite{gimbert2004can}, which is proved by induction on the number of edges of a game graph. This type of argument goes back to  a paper of Ehrenfeucht and Mycielski~\cite{ehrenfeucht1979positional}, where they provide an inductive proof of the positional determinacy of Mean Payoff Games.

This argument can be found in Section \ref{sec:ind}.

\item \emph{A fixed point argument.} Then we give a proof which uses a fixed point approach due to Shapley~\cite{shapley1953stochastic}. Shapley's technique is a standard way of establishing positional determinacy of Discounted Games. In this argument, one derives positional determinacy from the existence of a solution to a certain system of equations (sometimes called \emph{Bellman's equations}). In turn, to establish the existence of a solution, one uses Banach's fixed point theorem.

This argument can be found in Section \ref{sec:fixed}.

\item \emph{A strategy improvement argument.}  For Discounted Games, the existence of a solution to Bellman's equations can also be proved by \emph{strategy improvement}. This technique goes back to Howard~\cite{howard1960dynamic}; for its thorough treatment (as well as for its applications to other payoffs) we refer the reader to~\cite{fearnley2010strategy}. We generalize it to arbitrary continuous positionally determined payoffs.

This argument can be found in Section \ref{sec:str}.
\end{itemize}
The simplest way to obtain our main result is via the inductive argument (at the cost of appealing without a proof to the sufficient condition of Gimbert and Zielonka). In turn, two other proofs give the following additional results:
\begin{itemize}
\item Using the fixed point approach, in Section \ref{sec:fixed_app} we give an explicit description of the set of continuous positionally determined payoffs. Namely, it turns out that all continuous positionally determined payoffs are, in a sense, non-affine multi-discounted payoffs. We use this to give an example of a positionally determined payoff which does not reduce to multi-discounted payoffs in an ``algorithmic sense''.

\item Using the strategy improvement approach, in Section \ref{sec:alg} we show that a problem of finding a pair of optimal positional strategies is solvable in randomized subexponential time for any continuous positionally determined payoff.
\end{itemize}
 We also believe that our paper makes a useful addition to these approaches from a technical viewpoint.  For example, the main problem for the fixed point approach is to identify a metric with which one can carry out the same ``contracting argument'' as in the case of multi-discounted payoffs. To solve it, we obtain a result of independent interest about compositions of continuous functions. As for the strategy improvement approach, our main contribution is a generalization of such well-established tools as ``modified costs'' and a ``potential transformation lemma''~\cite[Lemma 3.6]{hansen2013strategy}.

Finally, we study continuous payoffs that are positional in \emph{stochastic games}. Namely, in Section \ref{sec:mdp} we show that any continuous payoff which is positional in Markov Decision Processes is multi-discounted. On the other hand, it is classical multi-discounted games are positional even in two-player stochastic games. Using it, we disprove the following conjecture of Gimbert~\cite{gimbert2007pure}: ``Any payoff function which is positional for the class of non-stochastic one-player
games is positional for the class of Markov decision processes''.

\section{Preliminaries}
\label{sec:prelim}

\subsection{Notation}
\label{subsec:notation}
We denote the function composition by $\circ$.
For two sets $A$ and $B$ by $B^A$ we denote the set of all functions from $A$ to $B$. 
We write $C = A\sqcup B$ for three sets $A, B, C$ if $A$ and $B$ are disjoint and $C = A\cup B$.

Take any set $A$. By $A^*$ we denote the set of all finite words over the alphabet $A$. By $A^+$ we denote the set of all non-empty finite words pver the alphabet $A$. Finally, by $A^\omega$ we denote the set of all infinite infinite words over the alphabet $A$. For $w\in A^*$, we let $|w|$ be the length of $w$. For $\alpha\in A^\omega$ we define $|\alpha| = \infty$.

For $u \in A^*$ and $v\in A^* \cup A^\omega$ we let $uv$ denote the concatenation of $u$ and $v$. We call $u\in A^*$ a prefix of $v\in A^*\cup A^\omega$ if for some $w\in A^*\cup A^\omega$ we have $v = uw$. 
 For $u\in A^*$, by $uA^\omega$ we denote the set $\{u \alpha \mid \alpha \in A^\omega\}$. Alternatively, $uA^\omega$ is the set of all $\beta \in A^\omega$ such that $u$ is a prefix of $\beta$.

For $u\in A^*$ and $k\in\mathbb{N}$ we define
\[u^k = \underbrace{u u \ldots u}_{k \mbox{ times}}.\]
In turn, if $u\in A^+$, we let $u^\omega\in A^\omega$ be a infinite word obtained by repeating $u$ infinitely many times. 
 We call $\alpha \in A^\omega$ ultimately periodic if $\alpha = u v^\omega$ for some $u\in A^*, v\in A^+$.

\subsection{Deterministic infinite duration games on finite directed graphs}
\label{subsec:idg}

\begin{defi}
Let $A$ be a finite set. A tuple  $G = \langle V, V_\Max, V_\Min, E\rangle$ is called an \emph{$A$-labeled game graph} if the following conditions hold:
\begin{itemize}
    \item $V, V_\Max, V_\Min, E$ are four finite sets with $V = V_\Max\sqcup V_\Min$, $E\subseteq V\times A\times V$;
    \item for every $s\in V$ there exist $a\in A$ and $t\in V$ such that $(s, a, t)\in E$.
\end{itemize}
\end{defi}
Elements of $V$ are called nodes of $G$. Nodes from $V_\Max$ (resp., $V_\Min$) are called Max's nodes (resp., Min's nodes). Elements of $E$ are called edges of $E$. For an edge $e = (s, a, t)\in E$ we define $\source(e) = s, \lab(e) = a, \target(e) = t$. We imagine $e$ as an arrow from $\source(s)$ to $\target(e)$ with the label $\lab(a)$.

We will apply the function $\lab$ not only to individual edges, but also to arbitrary finite or infinite sequences of edges. Namely, given a sequence of edges, we first apply $\lab$ to its elements, and then concatenate the resulting letters from $A$ in the same order as in the sequence. We will get a word over $A$ of the same length as the initial sequence of edges.

The out-degree of a node $v\in V$ is the number of $e\in E$ with $\source(e) = u$. The last requirement in the definition of an $A$-labeled game graph means that the out-degree of every node must be positive.

A path in $G$ is a non-empty (finite or infinite) sequence of edges of $G$ with a property that $\target(e) = \source(e^\prime)$ for any two consecutive edges $e$ and $e^\prime$ from the sequence. For a path $p$, we define $\source(p) = \source(e)$, where $e$ is the first edge of $p$. For a finite path $p$, we define $\target(p) = \target(e^\prime)$, where $e^\prime$ is the last edge of $p$. 

For technical convenience, we also consider $0$-length paths. Namely, for every node $s\in V$ we consider a $0$-length path $\lambda_s$, for which we define $\source(\lambda_s) = \target(\lambda_s) = s$. Hence, there are $|V|$ different $0$-length paths.

 If $p$ and $q$ are two paths of positive length and $p$ is finite, then we can consider their concatenation $pq$. This will be a path if and only if $\target(p) = \source(q)$.

Now, if $p = \lambda_s$ is a $0$-length path, then $\lambda_s q$ is a path if and only if $\source(q) = s$. In this case, $\lambda_s q = q$. Similarly, if $q = \lambda_s$ is a $0$-length path, then $p\lambda_s$ is a path if and only if $\target(p) = s$. In this case, $p\lambda_s = p$.

\medskip

Fix a finite set $A$ and an $A$-labeled game graph $G = \langle V, V_\Max, V_\Min, E\rangle$.
Consider the following infinite-duration game (IDG for short) which is played over $G$.
Players are called $\Max$ and $\Min$. Positions of the game are finite paths in $G$ (informally, these are possible finite developments of the game). Possible starting positions are paths of length $0$. Positions from where Max (resp., Min) is the one to move are finite paths with $\target(p)\in V_\Max$ (resp., $\target(p)\in V_\Min$).

The set of moves available at a position $p$ is the set $\{e\in E\mid \source(e) = \target(p)\}$ of edges that come out of the endpoint of $p$. A move $e$ from a position $p$ leads to a position $pe$.

A Max's strategy $\sigma$ in a game graph $G$ is a mapping, assigning to every position $p$ with $\target(p)\in V_\Max$ some move available at $p$. Similarly, a Min's strategy $\tau$ in a game graph $G$ is a mapping, assigning to every position $p$ with $\target(p)\in V_\Min$ some move available at $p$.

Let $\mathcal{P} = e_1 e_2 e_3\ldots$ be an infinite path in $G$. We say that $\mathcal{P}$ is \emph{consistent} with a Max's strategy $\sigma$ if for every finite prefix $p$ of $\mathcal{P}$ with $\target(p)\in V_\Max$ it holds that $\sigma(p)$ is the next edge of $\mathcal{P}$ after $p$. 
For $s\in V$ and for a Max's strategy $\sigma$ we let $\Cons(s, \sigma)$ be a set of all infinite paths in $G$ that start in $s$ and are consistent with $\sigma$. We use a similar terminology and notation for strategies of Min.

Given a Max's strategy $\sigma$, a Min's strategy $\tau$ and $s\in V$, \emph{the play of $\sigma$ and $\tau$ from $s$} is an infinite path $\mathcal{P}^{\sigma,\tau}_s$ which can be obtained as follows. First, set $p = \lambda_s$. Then repeat the following infinitely many times. If $\target(p)\in V_\Max$, extend it by the edge $\sigma(p)$. Similarly, if $\target(p)\in V_\Min$,  extend it by the edge $\tau(p)$. The resulting infinite path will be $\mathcal{P}^{\sigma,\tau}_s$. It is not hard to see that $\mathcal{P}^{\sigma,\tau}_s$ is a unique element of the intersection $\Cons(s, \sigma) \cap \Cons(s,\tau)$.

\medskip A Max's strategy $\sigma$ in an $A$-labeled game graph $G = \langle V, V_\Max,V_\Min, E\rangle$ is called \emph{positional} if $\sigma(p) = \sigma(q)$ for all finite paths $p$ and $q$ in $G$ with $\target(p) = \target(q) \in V_\Max$. For a positional strategy $\sigma$ of Max and for $u\in V_\Max$, we let $\sigma(u)$ be the move of $\sigma$ from any position whose endpoint is $u$. That is, we can view a positional Max's strategy $\sigma$ as a function $\sigma\colon V_\Max\to E$. Obviously, this function satisfies $\source(\sigma(u)) = u$ for all $u\in V_\Max$.
We define Min's positional strategies analogously.

We call an edge $e\in E$ \emph{consistent} with a Max's positional strategy $\sigma$ if $\source(e)\in V_\Max\implies e = \sigma(\source(e))$. We denote the set of edges that are consistent with $\sigma$ by $E^\sigma$. If $\tau$ is a Min's positional strategy, then we say that an edge $e\in E$ is consistent with $\tau$ if  $\source(e)\in V_\Min\implies e = \tau(\source(e))$. The set of edges that are consistent with a Min's positional strategy $\tau$ is denoted by $E_\tau$.

\medskip

Fix a finite set $A$ and an $A$-labeled game graph $G = \langle V, V_\Max, V_\Min, E\rangle$. Take any bounded function $\vphi\colon A^\omega\to \mathbb{R}$ to which we will refer to as a \emph{payoff} (``bounded'' means that $\vphi(A^\omega)\subseteq [-C, C]$ for some $C > 0$).
Given a Max's strategy $\sigma$ in $G$, its \emph{value} in a node $s\in V$ (w.r.t.~$\vphi$) is defined as follows:
\[\Val[\sigma](s) = \inf \vphi \circ\lab\big(\Cons(s, \sigma)\big).\]
That is, we first take all infinite paths from $s$ that are consistent with $\sigma$. Then we consider all infinite words over $A$ that are ``written'' on these paths. The set of these words is $\lab\big(\Cons(s, \sigma)\big)$. Finally, we take the infimum of our payoff over this set.

Similarly, if $\tau$ is a Min's strategy in $G$, then the value of $\tau$ in a node $s\in V$ (w.r.t.~$\vphi$) is the following quantity:
\[\Val[\tau](s) = \sup \vphi \circ\lab\big(\Cons(s, \tau)\big).\]
Observe that for any Max's strategy $\sigma$, for any Min's strategy $\tau$ and for any $s\in V$ we have:
\[\Val[\sigma](s) \le \vphi\circ\lab\big(\mathcal{P}^{\sigma,\tau}_s\big) \le \Val[\tau](s).\]
A Max's strategy $\sigma$ is called \emph{(uniformly) optimal} if $\Val[\sigma](s) \ge \Val[\sigma^\prime](s)$ for any $s\in V$ and for any Max's strategy $\sigma^\prime$. Similarly, a Min's strategy $\tau$ is called \emph{(uniformly) optimal} if $\Val[\tau](s) \le \Val[\tau^\prime](s)$ for any $s\in V$ and for any Min's strategy $\tau^\prime$.
\begin{rem}
``Uniformity'' here refers to the fact that a strategy is optimal irrespectively of the starting node. One, of course, could consider strategies that are optimal for some nodes but not for the other. However, this kind of optimality is out of scope of this paper. Thus, from now on, we write ``optimal strategies'' instead of ``uniformly optimal strategies''.
\end{rem}

A pair $(\sigma,\tau)$ of a Max's strategy $\sigma$ and a Min's strategy $\tau$ is called an \emph{equilibrium} if $\Val[\sigma](s) = \Val[\tau](s)$ for every $a\in V$.  It is easy to see
 that any strategy appearing in an equilibrium is optimal. On the other hand, if at least one equilibrium exists, then the following holds: the Cartesian product of the set of  optimal strategies of Max and  the set of  optimal strategies of Min is the set of equilibria. We say that $\vphi$ is \emph{determined} if in every $A$-labeled game graph there exists an equilibrium with respect to $\vphi$. We say that $\vphi$ is \emph{positionally determined} if every $A$-labeled game graph contains an equilibrium (w.r.t.~$\vphi$) of two positional strategies.

\begin{prop}
\label{conserv_eq}
If $A$ is a finite set, $\vphi\colon A^\omega\to\mathbb{R}$ is a payoff and $g\colon \vphi(A^\omega) \to\mathbb{R}$ is a bounded non-decreasing\footnote{Throughout the paper, we call a function $f\colon S\to\mathbb{R}, S\subseteq \mathbb{R}$ non-decreasing if for all $x, y\in S$ we have $x\le y \implies f(x) \le f(y)$.} function, then any equilibrium w.r.t.~$\vphi$ is also an equilibrium w.r.t.~$g\circ \vphi$.
\end{prop}
\begin{proof}
Let $(\sigma, \tau)$ be an equilibrium w.r.t.~$\vphi$, where $\sigma$ is a Max's strategy and $\tau$ is a Min's strategy. Our goal is to show that $(\sigma,\tau)$ is also an equilibrium w.r.t.~$g\circ\vphi$.

By definition, the values of $\sigma$ and $\tau$ w.r.t.~$\vphi$ coincide. That, for every node $s$, we have:
\[\inf \vphi \circ\lab\big(\Cons(s, \sigma)\big) = \sup \vphi \circ\lab\big(\Cons(s, \tau)\big).\]
We have to derive from this that the values of $\sigma$ and $\tau$ w.r.t.~$g\circ \vphi$ also coincide. That is, we have to show that:
\begin{equation}
\label{eq:www}
\inf g\circ\vphi \circ\lab\big(\Cons(s, \sigma)\big) = \sup g\circ\vphi \circ\lab\big(\Cons(s, \tau)\big)
\end{equation}
for every node $s$. The sets   $\vphi \circ\lab\big(\Cons(s, \sigma)\big)$ and $\vphi \circ\lab\big(\Cons(s, \tau)\big)$ have a common element $\vphi\circ\lab\big(\mathcal{P}^{\sigma,\tau}_s\big)$. Since the infimum of the first set equals the supremum of the second set, their common element $ \vphi\circ\lab\big(\mathcal{P}^{\sigma,\tau}_s\big)$ must be the minimum of the first set and the maximum of the second set. Due to the fact that the function $g$ is non-decreasing, we have that $g\big( \vphi\circ\lab\big(\mathcal{P}^{\sigma,\tau}_s\big)\big)$ is the minimum of $g\circ\vphi \circ\lab\big(\Cons(s, \sigma)\big)$ and the maximum $g\circ\vphi \circ\lab\big(\Cons(s, \sigma)\big)$. This implies \eqref{eq:www}.
\end{proof}
\begin{cor}
\label{non_decr_comp}
If $A$ is a finite set, $\vphi\colon A^\omega\to\mathbb{R}$ is a positionally determined payoff and $g\colon \vphi(A^\omega) \to\mathbb{R}$ is a bounded non-decreasing function, then $g\circ \vphi$ is a positionally determined payoff.
\end{cor}

\subsection{Continuous payoffs}

For a finite set $A$, we consider the set $A^\omega$ as a topological space. Namely, we take the discrete topology on $A$ and the corresponding product topology on $A^\omega$. In this product topology, open sets are sets of the form
\[\mathcal{S} = \bigcup\limits_{u\in S} u A^\omega,\]
where $S\subseteq A^*$. When we say that a payoff $\vphi\colon A^\omega\to\mathbb{R}$ is \emph{continuous} we always mean continuity with respect to this product topology (and with respect to the standard topology on $\mathbb{R}$). The following proposition gives a convenient way to establish continuity of payoffs.
\begin{prop}
\label{cont_help}
Let $A$ be a finite set.
A payoff $\varphi\colon A^\omega \to \mathbb{R}$ is continuous if and only if for any $\alpha\in A^\omega$ and  for any infinite sequence $\{\beta_n\}_{n = 1}^\infty$ of elements of $A^\omega$ the following holds. If for all $n\ge 1$ it holds that $\alpha$ and $\beta_n$ have the same prefixes of length $n$, then $\lim\limits_{n \to \infty} \varphi(\beta_n)$ exists and equals $\varphi(\alpha)$.
\end{prop}
\begin{proof}
First, assume that $\vphi$ is continuous. Take any $\varepsilon > 0$. We have to show that for some $n_0$ it holds that $\vphi(\beta_n) \in (\vphi(\alpha)-\varepsilon, \vphi(\alpha) +\varepsilon)$ for all $n\ge n_0$.  The set $\vphi^{-1}((\vphi(\alpha)-\varepsilon, \vphi(\alpha) +\varepsilon))$ must be open. So for some $S\subseteq A^*$ we have:
\[\vphi^{-1}((\vphi(\alpha)-\varepsilon, \vphi(\alpha) +\varepsilon)) = \bigcup\limits_{u\in S} u A^\omega.\]
Since obviously $\alpha \in \vphi^{-1}((\vphi(\alpha)-\varepsilon, \vphi(\alpha) +\varepsilon))$, there exists $u\in S$ such that $\alpha \in u A^\omega$. Hence for $n\ge |u|$ we have $\beta_n \in u A^\omega \subseteq\vphi^{-1}((\vphi(\alpha)-\varepsilon, \vphi(\alpha) +\varepsilon))$, as required.

Let us now establish the other direction of the proposition. It is enough to show that for any $x, y \in\mathbb{R}$ with $x < y$ the set $\vphi^{-1}((x, y))$ is open. Take any $\alpha \in \vphi^{-1}((x, y))$. Let us show that there exists $n(\alpha)$ such that all $\beta\in A^\omega$ that coincide with $\alpha$ in the first $n(\alpha)$ elements belong to $\vphi^{-1}((x, y))$. Indeed, otherwise for any $n$ there exists $\beta_n$, coinciding with $\alpha$ in the first $n$ elements, such that $\beta_n \notin\vphi^{-1}((x, y))$. Now, the limit $\lim_{n\to\infty}\vphi(\beta_n)$ must exist and must be equal to $\vphi(\alpha)$. But $\vphi(\alpha) \in (x, y)$ and all $\vphi(\beta_n)$ are not in this interval, contradiction.

Now, for $\alpha \in \vphi^{-1}((x, y))$ let $u_\alpha \in A^{n(\alpha)}$ be the $n(\alpha)$-length prefix of $\alpha$. Observe that
\[\vphi^{-1}((x, y)) = \bigcup\limits_{\alpha \in \vphi^{-1}((x, y))} u_\alpha A^\omega.\]
So the set $\vphi^{-1}((x, y))$ is open, as required. \qedhere

\end{proof}
 
For a finite set $A$, the space $A^\omega$ is compact by Tychonoff's theorem\footnote{Compactness of $A^\omega$ can also be shown via K\H{o}nig's lemma.}. This has the following consequence which is important for this paper: if $\vphi\colon A^\omega\to\mathbb{R}$ is a continuous payoff, then $\vphi(A^\omega)$ is a compact subset of $\mathbb{R}$.

\subsection{MDPs}

This subsection concerns stochastic games, but we deal with them only in Section \ref{sec:mdp}. So for the rest of our results, one can skip this subsection.

In fact, we will need only one-player stochastic games, also known as Markov Decision Processes (MDPs). We will follow a formalization of Gimbert~\cite{gimbert2007pure}.

We use the following notation.
Let $A$ be a finite set. By $\mathfrak{S}^{bor}_A$ we mean the Borel $\sigma$-algebra on $A^\omega$ (generated by the product topology from the previous subsection). By $\Delta(S)$ we denote the set of all probability distributions over a finite set $S$.
\begin{defi} 
\label{def:mdp}
Let $A$ be a finite set. An \emph{$A$-labeled MDP} is a tuple $\mathcal{M} = \langle S, \act, \lab\rangle$, where 
\begin{itemize}
\item $S$ is a finite set of \emph{states} of $\mathcal{M}$;
\item $\act\subseteq S\times \Delta(S)$ is a finite set of \emph{actions} of $\mathcal{M}$;
\item $\lab\colon \act\times S\to A$ is the \emph{labeling function} of $\mathcal{M}$;
\end{itemize}
such that for every $s\in S$ there exists $P\in \Delta(S)$ such that $(s, P) \in \act$.
\end{defi}

Given an $A$-labeled MDP  $\mathcal{M} = \langle S, \act, \lab\rangle$, we imagine that there is a single player called Max traveling over the states of $\mathcal{M}$. When Max is in a state $s\in S$, he considers the set of all actions of $\mathcal{M}$ whose first coordinate is $s$ (by definition, this set is non-empty for every $s\in S$). He chooses one such action $(s, P)$. Then Max samples his next location according to $P$. This continues for infinitely many turns.

For $e = (s, P)\in \act$, we define $\source(e) = s$ and $\dist[e] = P$.

The set $T = \act\times S$, which is the domain of the function $\lab$, is called the set of \emph{transitions} of $\mathcal{M}$. Informally, transitions describe what happens in one turn. Namely, a transition $(e,  s) \in\act\times S$ means that in the beginning of a turn, Max was in the state $\source(e)$, then he took the action $e$, and this led him to the state $s$.

Consistent sequences of transitions are called \emph{histories}. Namely, a non-empty sequence $h = (e_1, s_1) (e_2, s_2) (e_3, s_3)\ldots \in T^+\cup T^\omega$ is called a history if for every $2\le i \le |h|$, we have $s_{i-1} = \source(e_{i})$. We set $\source(h) = \source(e_1)$ and, if $h$ is finite, $\target(h) = s_{|h|}$. We also map each state $s\in S$ to a $0$-length history $\lambda_s$ with $\source(\lambda_s) = \target(\lambda_s) = s$. These histories correspond to $|S|$ possible starting positions of Max.

A strategy $\sigma$ of Max is a mapping, which to every finite history $h$ assigns an action $\sigma(h)\in\act$ such that $\target(h) = \source(\sigma(h))$. Informally, $\sigma(h)$ is the action which, according to $\sigma$, Max takes after $h$.

Given $s\in S$, a strategy $\sigma$ defines a function $P^\sigma_s\colon T^*\to [0, 1]$. Informally, $P^\sigma_s(h)$ is the probability that we will see a history $h$ if Max starts in $s$ and plays according to $\sigma$. It can be defined inductively.

First, we set $P^\sigma_s(\mbox{empty word}) = 1$. The empty word here corresponds to the initial history $\lambda_s$. Now, given $(e, s_1) \in T$, we set $P^\sigma_s((e, s_1)) = 0$ if $e\neq \sigma(\lambda_s)$ and $P^\sigma_s((e, s_1)) = \dist[e](s_1)$ if $e = \sigma(\lambda_s)$. That is, if $e$ is not the action played by Max according to $\sigma$ in the starting position, then the probability of the transition $(e, s_1)$ is $0$. Otherwise, the probability of $(e, s_1)$ is the probability that the action $e = \sigma(\lambda_s)$ brings us to $s_1$.

More generally, assume that $P^\sigma_s$ is already defined for all $h\in T^1\cup T^2\ldots \cup T^n$. Take any $h = (e_1, s_1) \ldots (e_n, s_n) (e_{n + 1}, s_{n + 1})$. If $h$ is not a history, or if $\source(h) \neq s$, we set $P^\sigma_s(h) = 0$. Similarly, if $e_{n + 1} \neq \sigma((e_1, s_1) \ldots (e_n, s_n))$, that is, if $e_{n+1}$ is not the action played by $\sigma$ after $(e_1, s_1) \ldots (e_n, s_n)$, then we set $P^\sigma_s(h) = 0$. Finally, if $h$ is a history with $\source(h) = s$, and if $e_{n + 1} = \sigma((e_1, s_1) \ldots (e_n, s_n))$, then we set 
\[P^\sigma_s(h) = P^\sigma_s((e_1, s_1) \ldots (e_n, s_n)) \cdot \dist[e_{n+1}](s_{n + 1}).\]
Obviously, we have $P^\sigma_s(h) = \sum_{t\in T} P^\sigma_s(ht)$ for every $h\in T^*$. Hence, by the Caratheodory's extension theorem, there is a unique probability measure $\mathcal{P}^\sigma_s$ on $\mathfrak{S}^{bor}_T$ such that
\[\mathcal{P}^\sigma_s(hT^\omega) = P^\sigma_s(h)\]
for every $h\in T^*$. Intuitively, $\mathcal{P}^\sigma_s$ is the probability distribution over infinite histories, generated by playing $\sigma$ for infinitely many turns, starting from $s$.

\medskip

Any sequence of transitions $h \in T^*\cup T^\omega$ can be mapped to a word $\lab(h)$ over the set of labels by setting:
\[\lab(h) = \lab(h_1) \lab(h_2)\lab(h_3)\ldots\]
Now, fix a payoff function $\vphi\colon A^\omega\to\mathbb{R}$. It maps any infinite history $\mathcal{H} \in T^\omega$ to its reward, defined as $\vphi(\lab(\mathcal{H}))$. Max wants a strategy which maximizes the expected value of the reward. That is, Max wants to attain
\begin{equation}
\label{mdp_max}
\mathbb{E}_{\mathcal{H}\sim\mathcal{P}^\sigma_s} \left[\vphi\circ\lab(\mathcal{H})\right]\to \max
\end{equation}
over his strategies $\sigma$, for all $s\in S$.
This expectation is well defined if $\vphi\circ \lab$ is bounded and measurable with respect to $\mathfrak{S}_T^{bor}$. Since $\lab\colon T^\omega\to A^\omega$ is continuous, it is well-defined if $\vphi\colon A^\omega\to\mathbb{R}$ is bounded and measurable with respect to $\mathfrak{S}_A^{bor}$.

For brevity, we will abbreviate the expectation in \eqref{mdp_max} by
\[\mathbb{E} \left[\varphi \circ\lab\big(\mathcal{P}_s^\sigma\big)\right].\]
\begin{defi}
Let $A$ be a finite set and $\vphi\colon A^\omega\to\mathbb{R}$ be a bounded measurable payoff.

We say that a strategy $\sigma$ in an $A$-labeled MDP $\mathcal{M}$ is \emph{optimal} if for any strategy $\sigma^\prime$ and for any state $s$ of $\mathcal{M}$ we have:
\[\mathbb{E}  \left[\varphi \circ\lab\big(\mathcal{P}_s^\sigma\big)\right]\ge \mathbb{E} \left[ \varphi \circ\lab\big(\mathcal{P}_s^{\sigma^\prime}\big)\right].\]

We say that a strategy $\sigma$ in an $A$-labeled MDP $\mathcal{M}$ is \emph{positional} if for any two finite histories $h_1$ and $h_2$ in $\mathcal{M}$ we have $\target(h_1) = \target(h_2) \implies \sigma(h_1) = \sigma(h_2)$.

We say that $\vphi$ is \emph{positionally determined in MDPs} if every $A$-labeled MDP has an optimal positional strategy w.r.t.~$\vphi$.
\end{defi}
In the paper, we will use this definition only for continuous $\vphi$ -- they all are, of course, bounded (by compactness of $A^\omega$) and measurable.

\section{Statement of the Main Result and its \texorpdfstring{``}{“}Only If\texorpdfstring{''}{”} Part}
\label{sec:results}
Our main result establishes a simple property which is equivalent to positional determinacy for continuous payoffs.

\begin{defi}
\label{def:prefix_monotone}
Let $A$ be a finite set. A payoff $\vphi\colon A^\omega\to\mathbb{R}$ is called \textbf{prefix-monotone} if there are no $u, v\in A^*$, $\beta,\gamma\in A^\omega$ such that
$\vphi(u\beta) > \vphi(u\gamma) \mbox{ and } \vphi(v\beta) < \vphi(v\gamma)$.
\end{defi}

(One can note that prefix-independence trivially implies prefix-monotonicity. On the other hand, no prefix-independent payoff is continuous, unless it takes just 1 value.)

\begin{thm}
\label{thm:char}
Let $A$ be a finite set and $\varphi\colon A^\omega \to \mathbb{R}$ be a  continuous payoff. Then $\varphi$ is positionally determined if and only if $\varphi$ is prefix-monotone.
\end{thm}
The fact that any continuous positionally determined payoff must be prefix-monotone
is proved below in this section.  Three different proofs of the ``if'' part of Theorem \ref{thm:char} are given in, respectively, Sections \ref{sec:ind}, \ref{sec:fixed} and \ref{sec:str}. As an illustration of our result, we first give a formal definition of multi-discounted payoffs and show that they are continuous and prefix-monotone.
\begin{defi} Let $A$ be a finite set. Then 
a payoff $\vphi\colon A^\omega\to\mathbb{R}$ is \textbf{multi-discounted} if there are functions $\lambda\colon A\to[0, 1)$ and $w\colon A\to\mathbb{R}$ such that
\begin{equation}
\label{multi_disc}
\vphi(a_1 a_2 a_3\ldots) = \sum\limits_{n = 1}^\infty \lambda(a_1) \cdot \ldots \cdot \lambda(a_{n-1}) \cdot w(a_n)
\end{equation}
for all $a_1 a_2 a_3\ldots\in A^\omega$.
\end{defi}

\begin{prop} All multi-discounted payoffs are continuous and prefix-monotone. 
\end{prop}
\begin{proof}
Let $A$ be a finite set and $\vphi\colon A^\omega\to\mathbb{R}$ be a multi-discounted payoff, defined by $\lambda\colon A\to[0, 1)$ and $w\colon A\to\mathbb{R}$. Take any $W > 0$ such that $\lambda(a) < 1 - \frac{1}{W}$ and $|w(a)| < W$ for every $a\in A$.

Let us first show that $\vphi$ is continuous. Take any $\alpha, \beta\in A^\omega$ that coincide in the first $n$ elements. It is sufficient to bound the difference $|\vphi(\alpha) - \vphi(\beta)|$ by some quantity which depends only on $n$ and tends to $0$ as $n\to\infty$.
 First, observe that the value of $\vphi$ never exceeds $W\cdot \frac{1}{1 - (1 - \frac{1}{W})} = W^2$. Now, let $u = a_1 a_2 \ldots a_n\in  A^n$ be the first $n$ letters of $\alpha$ and $\beta$. Then $\alpha = u\alpha^\prime, \beta = u\beta^\prime$ for some $\alpha^\prime, \beta^\prime\in A^\omega$. It is not hard to derive from \eqref{multi_disc} that:
\begin{equation}
\label{append_disc}
\vphi(\alpha) - \vphi(\beta) = \vphi(u\alpha^\prime) -  \vphi(u\beta^\prime) = \lambda(a_1)\cdot\ldots\cdot\lambda(a_n) \cdot (\vphi(\alpha^\prime) - \vphi(\beta^\prime)).
\end{equation}
This means that the difference $|\vphi(\alpha) - \vphi(\beta)|$ is bounded by $(1 - \frac{1}{W})^n \cdot W^2$. This quantity tends to $0$ as $n\to \infty$. Hence, $\vphi$ is continuous.

Equation \eqref{append_disc} also implies that $\vphi$ is prefix-monotone. Indeed, it gives that for any $u\in A^*$ and $\beta,\gamma\in A^\omega$ there exists $\lambda \ge 0$ such that  $\vphi(u\beta) - \vphi(u\gamma) = \lambda\cdot(\vphi(\beta) - \vphi(\gamma))$. This equality gives us that:
\[\vphi(u\beta) > \vphi(u\gamma) \implies \vphi(\beta) >\vphi(\gamma), \qquad \vphi(u\beta) < \vphi(u\gamma) \implies \vphi(\beta) <\vphi(\gamma).\]
Hence, there are no $u,v\in A^*$ and $\beta,\gamma\in A^\omega$ such that
$\vphi(u\beta) > \vphi(u\gamma)$ and $\vphi(v\beta) < \vphi(v\gamma)$.
\end{proof}

\begin{proof}[Proof of the ``only if'' part of Theorem \ref{thm:char}]
Assume that $\varphi$ is not prefix-monotone. Then for some $u, v\in A^*$  and $\alpha, \beta \in A^\omega$ we have 
\begin{equation}
\label{pm_violate}
\vphi(u\alpha) > \vphi(u\beta) \mbox{ and } \vphi(v\alpha) < \vphi(v\beta).
\end{equation}
 First, notice that by the continuity of $\vphi$ we may assume that $\alpha$ and $\beta$ are ultimately periodic. Indeed, take any $a\in A$ and for every $n\in\mathbb{N}$, define $\alpha^n, \beta^n\in A^\omega$ as follows:
\[\alpha^n = \alpha_1 \alpha_2\ldots \alpha_n a^\omega, \qquad \beta^n = \beta_1 \beta_2\ldots\beta_n a^\omega.\]
By  continuity of $\vphi$, we have:
\[\lim\limits_{n\to \infty} \vphi(u\alpha^n) = \vphi(u\alpha), \qquad \lim\limits_{n\to \infty} \vphi(v\alpha^n) = \vphi(v\alpha),\]
\[\lim\limits_{n\to \infty} \vphi(u\beta^n) = \vphi(u\beta), \qquad \lim\limits_{n\to \infty} \vphi(v\beta^n) = \vphi(v\beta).\]
These equations imply that if $u, v, \alpha, \beta$ violate prefix-monotonicity, then so do $u, v, \alpha^n, \beta^n$ for some $n\in\mathbb{N}$ (and $\alpha^n, \beta^n$ are ultimately periodic for every $n$).

Thus, we assume from now on that $\alpha, \beta$ are ultimately periodic. Then $\alpha = p (q)^\omega$ and $\beta = w(r)^\omega$ for some $p, q, w, r\in A^*$. Consider an $A$-labeled game graph from Figure \ref{gg} (all its nodes are owned by $\Max$).
\begin{figure}[h!]
\centering
\begin{tikzpicture}[gamegraph]
   
\node[draw, circle] (a) {$a$};
\node[draw, circle, below = 2cm of a] (b) {$b$};
\node[draw, circle, below right = 1cm  and 2 cm of a] (c) {$c$};
\node[draw, circle, right = 4cm of b] (d) {\textcolor{white}{$a$}};
\node[draw, circle, right = 4cm of a] (e) {\textcolor{white}{$a$}};
%

\draw[transition] (a) -- (c) node[midway, above] {$x$};
\draw[transition] (b) -- (c) node[midway, below] {$y$};
\draw[transition] (c) -- (d) node[midway, below] {$u$};
\draw[transition] (c) -- (e) node[midway, above] {$p$};


\path[transition,overlay]
  (e) edge [in=330,out=30,out distance=3cm,in distance=3cm] node[right] {$q$} (e);

\path[transition,overlay]
  (d) edge [in=330,out=30,out distance=3cm,in distance=3cm] node[right] {$v$} (d);


\end{tikzpicture}
  \caption{A game graph where $\vphi$ is not positionally determined.}
\label{gg}
\end{figure}
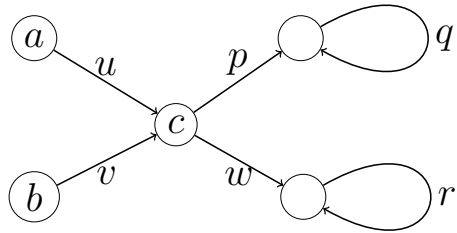

There are two positional strategies of Max in this game graph, one which goes along $p$ from $c$, and the other which goes along $w$ from $c$. The first one is not optimal when the game starts in $b$, and the second one is not optimal when the game starts in $a$ (because of \eqref{pm_violate}). So $\vphi$ is not positionally determined in this game graph.
\end{proof}
\begin{rem}
In this argument, it is crucial that our definition of positional determinacy is ``uniform''. That is, we require that some positional strategy is optimal for all the nodes. Allowing  each starting node to have its own optimal positional strategy gives us a weaker, ``non-uniform'' version of positional determinacy. It is not clear whether non-uniform positional determinacy implies prefix-monotonicity for continuous payoffs. At the same time, we are not even aware of a payoff which is positional in the non-uniform sense, but not in the uniform sense.
\end{rem}

\section{Inductive Argument}
\label{sec:ind}

In this section, we show that any continuous prefix-monotone payoff is positionally determined, using the following sufficient condition due to Gimbert and Zielonka~\cite[Theorem 1]{gimbert2004can}:
\begin{prop}
\label{fairly_mixing}
Let $A$ be a finite set. Any payoff $\vphi\colon A^\omega\to\mathbb{R}$, satisfying the following three conditions:
\begin{itemize}
\item \textbf{(a)} for all $u\in A^*$ and $\alpha,\beta\in A^\omega$ we have that
$\vphi(\alpha) \le \vphi(\beta) \implies \vphi(u\alpha) \le \vphi(u\beta)$;
\item \textbf{(b)} for all $u\in A^+$ and $\alpha \in A^\omega$ we have that
\[\min\{\vphi(u^\omega), \vphi(\alpha)\} \le \vphi(u\alpha) \le \max\{\vphi(u^\omega), \vphi(\alpha)\};\]
\item \textbf{(c)} for any infinite sequence  $\{x_n\in A^+\}_{n = 0}^\infty$ it holds that:
\begin{align*}
&\min\{\vphi(x_0 x_2 x_4\ldots), \vphi(x_1 x_3 x_5\ldots), \inf_{n\ge 0}\vphi(x_n^\omega)\}\le\vphi(x_0 x_1 x_2\ldots) \\
&\le\max\{\vphi(x_0 x_2 x_4\ldots), \vphi(x_1 x_3 x_5\ldots), \sup_{n\ge 0}\vphi(x_n^\omega)\}\\
\end{align*}
\end{itemize}
is positionally determined.
\end{prop}
We observe that in case of continuous payoffs, one can get rid of the conditions \emph{\textbf{(b)}} and \emph{\textbf{(c)}} in this Proposition. A weaker version of this statement was proved in the on-line version of~\cite{gimbert2004can}. Namely, it was shown there that one can get rid of the condition \emph{\textbf{(c)}} for continuous payoffs.

\begin{prop}
\label{agivesb}
For continuous payoffs, the condition \textbf{(a)} of Proposition \ref{fairly_mixing} implies the conditions \textbf{(b)} and \textbf{(c)} of Proposition \ref{fairly_mixing}.
\end{prop}
\begin{proof}
Take any finite set $A$ and any continuous payoff $\vphi\colon A^\omega\to\mathbb{R}$ satisfying the condition \emph{\textbf{(a)}} of Proposition \ref{fairly_mixing}. We first show that $\vphi$ satisfies the condition \emph{\textbf{(b)}} of this proposition.
We will only show that $\vphi(u\alpha) \le \max\{\vphi(u^\omega), \vphi(\alpha)\}$, the other inequality from this condition can be proved similarly. If $\vphi(u\alpha) \le \vphi(\alpha)$, then we are done. Assume now that $\vphi(u\alpha) > \vphi(\alpha)$. By repeatedly applying \emph{\textbf{(a)}}, we obtain $\vphi(u^{i+1}\alpha) \ge \vphi(u^i\alpha)$ for every $i\in\mathbb{N}$. In particular, for every $i\ge 1$ we get that $\vphi(u^i\alpha) \ge \vphi(u\alpha)$. By continuity of $\vphi$, we have that $\lim_{i\to\infty}\vphi(u^i\alpha) = \vphi(u^\omega)$. Hence, $\vphi(u^\omega) \ge \vphi(u\alpha)$.

Now we show that $\vphi$ satisfies the condition \emph{\textbf{(c)}} of Proposition \ref{fairly_mixing}. We will only show that $\vphi(x_0 x_1 x_2\ldots)\le\max\{\vphi(x_0 x_2 x_4\ldots), \vphi(x_1 x_3 x_5\ldots), \sup_{n\ge 0}\vphi(x_n^\omega)\}$, the other inequality from this condition has the same proof. Namely, we will show that if $\vphi(x_0 x_1 x_2\ldots) > \sup_{n\ge 0}\vphi(x_n^\omega)$, then $\vphi(x_0 x_1 x_2\ldots) \le \vphi(x_0x_2 x_4\ldots)$.   Note that this claim is stronger than we need.

 First, we show that $\vphi(x_n x_{n+1} x_{n+2}\ldots) \le \vphi(x_{n+1} x_{n + 2} x_{n+3}\ldots)$ for every $n\ge 0$. This can be easily proved by induction on $n$. Let us start with the induction base. By the condition \textbf{\emph{(b)}}, which is already established for $\vphi$, we have $\vphi(x_0 x_1 x_2\ldots) \le\max\{\vphi(x_0^\omega), \vphi(x_1 x_2 x_3\ldots)\}$. Since, $\vphi(x_0 x_1 x_2\ldots) > \vphi(x_0^\omega)$, we have $\vphi(x_0 x_1 x_2\ldots) \le \vphi(x_1 x_2 x_3\ldots)$.

Let us now perform the induction step. Assume that it is already proved that $\vphi(x_n x_{n+1} x_{n+2}\ldots) \le \vphi(x_{n+1} x_{n + 2} x_{n+3}\ldots)$ for all $n\le N$. In particular, this means that $\vphi(x_{N+1} x_{N+2} x_{N+3}\ldots) \ge \vphi(x_0 x_1 x_2\ldots) \ge \vphi(x_{N+1}^\omega)$. Then, by the same argument as in the induction base, we get $\vphi(x_{N+1} x_{N+2} x_{N+3}\ldots) \le\vphi(x_{N+2} x_{N+3} x_{N+4}\ldots)$.

We will now prove that $\vphi(x_0 x_2 \ldots x_{2n} x_{2n+1} x_{2n + 2}\ldots) \ge \vphi(x_0 x_1 x_2\ldots)$ for every $n\ge 0$. For $n = 0$, the left-hand side and the right-hand side coincide. Then we show that $\vphi(x_0 x_2 \ldots x_{2n} x_{2n+1} x_{2n + 2}\ldots) \le \vphi(x_0 x_2 \ldots x_{2n + 2} x_{2n+3} x_{2n + 4}\ldots)$ for every $n \ge 0$. Due to the condition \textbf{\emph{(b)}}, we have $\vphi(x_{2n + 1} x_{2n + 2} x_{2n + 3}\ldots) \le \max\{\vphi(x_{2n + 1}^\omega), \vphi(x_{2n + 2} x_{2n + 3} x_{2n + 4}\ldots)\}$. On the other hand, as we have shown, $\vphi(x_{2n+1} x_{2n + 2} x_{2n + 3} \ldots) \ge \vphi(x_0 x_1 x_2\ldots) > \vphi(x_{2n + 1}^\omega)$. Hence $\vphi(x_{2n + 1} x_{2n + 2} x_{2n + 3}\ldots) \le \vphi(x_{2n + 2} x_{2n + 3} x_{2n + 4}\ldots)$. It remains to apply \textbf{\emph{(a)}} by appending $x_0 x_2\ldots x_{2n}$ to both sides.

Thus, we have established
that  $\vphi(x_0 x_2 \ldots x_{2n} x_{2n+1} x_{2n + 2}\ldots) \ge \vphi(x_0 x_1 x_2\ldots)$ for every $n \ge 0$. By continuity of $\vphi$, the left-hand side of this inequality converges to $\vphi(x_0 x_2 x_4\ldots)$ as $n\to\infty$. Hence, we get that $\vphi(x_0 x_2 x_4\ldots) \ge \vphi(x_0 x_1 x_2\ldots)$, as required.
\end{proof}

Thus, to establish that some continuous payoff is positionally determined, it is enough to demonstrate that this payoff satisfies the condition \textbf{\emph{(a)}} of Proposition \ref{fairly_mixing}. Let us now reformulate this condition using the following definition.

\begin{defi}
\label{shift_det}
Let $A$ be a finite set. A payoff $\vphi\colon A^\omega\to\mathbb{R}$ is called \textbf{shift-deterministic} if for all $a\in A, \beta,\gamma\in A^\omega$ we have
$\vphi(\beta) = \vphi(\gamma) \implies \vphi(a\beta) = \vphi(a\gamma)$.
\end{defi}

\begin{clm}
\label{obv}
Let $A$ be a finite set. A payoff $\vphi\colon A^\omega\to\mathbb{R}$ satisfies the condition \textbf{\emph{(a)}} of Proposition \ref{fairly_mixing} if and only if $\vphi$ is prefix-monotone and shift-deterministic.
\end{clm}
\begin{proof}
Assume first that $\vphi$ satisfies the condition \textbf{\emph{(a)}} of Proposition \ref{fairly_mixing}. It is shift-deterministic, because
\begin{align*}
\vphi(\beta) = \vphi(\gamma) &\implies \vphi(\beta)\le \vphi(\gamma)\land \vphi(\gamma)\le \vphi(\beta)\\ &\implies \vphi(a\beta)\le \vphi(a\gamma)\land \vphi(a\gamma)\le \vphi(a\beta) \implies \vphi(a\beta) =  \vphi(a\gamma)
\end{align*}
for every $a\in A, \beta,\gamma\in A^\omega$. In turn, assume for contradiction that $\vphi$ is not prefix-monotone. Then $\vphi(u\beta) > \vphi(u\gamma)$ and $\vphi(v\beta) < \vphi(v\gamma)$ for some $u,v\in A^*$ and $\beta,\gamma\in A^\omega$. Due to the contraposition to the condition \textbf{\emph{(a)}} of Proposition \ref{fairly_mixing}, we have $\vphi(u\beta) > \vphi(u\gamma) \implies \vphi(\beta) > \vphi(\gamma)$ and $\vphi(v\gamma) > \vphi(v\beta) \implies \vphi(\gamma) > \vphi(\beta)$, contradiction. 

Now, assume that $\vphi$ is prefix-monotone and shift-deterministic. Take any $u\in A^*$ and $\alpha,\beta\in A^\omega$ such that $\vphi(\alpha)\le \vphi(\beta)$. We have to derive that $\vphi(u\alpha)\le \vphi(u\beta)$. If $\vphi(\alpha) =  \vphi(\beta)$, then $\vphi(u\alpha) = \vphi(u\beta)$ because $\vphi$ is shift-deterministic (we apply the definition of the shift-determinism to letters of $u$ from right to left). If $\vphi(\alpha) < \vphi(\beta)$, then $\vphi(u\alpha)\le \vphi(u\beta)$ because otherwise $\vphi$ is not prefix-monotone.
\end{proof}

The above discussion gives the following sufficient condition for positional determinacy.

\begin{prop}
\label{cpmsd}
Let $A$ be a finite set. Any continuous prefix-monotone shift-deterministic  payoff $\vphi\colon A^\omega\to\mathbb{R}$ is positionally determined.
\end{prop}

Still, some argument is needed for continuous prefix-monotone payoffs that are not shift-deterministic. To tie up loose ends, we prove the following:

\begin{prop}
\label{phi_to_psi}
Let $A$ be a finite set and let $\vphi\colon A^\omega\to\mathbb{R}$ be a continuous prefix-monotone payoff. Then $\vphi  = g\circ \psi$ for some continuous prefix-monotone shift-deterministic payoff $\psi\colon A^\omega\to\mathbb{R}$ and for some continuous\footnote{Throughout the paper we call a function $f\colon S\to\mathbb{R}$, $S\subseteq \mathbb{R}^n$ continuous if $f$ is continuous with respect to the restriction of the standard topology of $\mathbb{R}^n$ to $S$.} non-decreasing function $g\colon \psi(A^\omega) \to\mathbb{R}$ (note that since $g$ is defined on a compact and is continuous, it is also bounded).
\end{prop}

Due to Corollary \ref{non_decr_comp}, this proposition means that all continuous prefix-monotone payoffs are positionally determined. In fact, we do not need continuity of $g$ here, but it will be useful later. Thus, once we establish Proposition \ref{thm:char}, our first proof of Theorem \ref{thm:char} will be finished.

\begin{proof}[Proof of Proposition \ref{phi_to_psi}]
Define a payoff $\psi\colon A^\omega\to\mathbb{R}$ as follows:
\begin{equation}
\label{psi_def}
\psi(\gamma) = \sum\limits_{w\in A^*} \left(\frac{1}{|A| + 1}\right)^{|w|} \vphi(w\gamma), \qquad \gamma \in A^\omega.
\end{equation}
First, why is $\psi$ well-defined, i.e.,  why does this series converge? Since $A^\omega$ is compact, so is $\vphi(A^\omega) \subseteq \mathbb{R}$, because $\vphi$ is continuous. Hence, $\vphi(A^\omega) \subseteq [-W, W]$ for some $W > 0$, which means that \eqref{psi_def} is bounded by the following absolutely converging series:
\[\sum\limits_{w\in A^*} W\cdot\left(\frac{1}{|A| + 1}\right)^{|w|}.\]

We shall show that $\psi$ is continuous, prefix-monotone and shift-deterministic, and that $\vphi = g\circ \psi$ for some continuous non-decreasing $g\colon \psi(A^\omega)\to\mathbb{R}$.

\bigskip

\textbf{Why is $\psi$ continuous?}  We will use Proposition \ref{cont_help}. Consider any $\alpha \in A^\omega$ and any infinite sequence $\{\beta_n\}_{n\in\mathbb{N}}$ of elements of $A^\omega$ such that for all $n$, the words $\alpha$ and $\beta_n$ have the same prefix of length $n$. We have to show that $\psi(\beta_n)$ converges to $\psi(\alpha)$ as $n\to\infty$. By definition:
\[\psi(\beta_n) = \sum\limits_{w\in A^*} \left(\frac{1}{|A| + 1}\right)^{|w|} \vphi(w\beta_n),\qquad \psi(\alpha) = \sum\limits_{w\in A^*} \left(\frac{1}{|A| + 1}\right)^{|w|} \vphi(w\alpha).\]
For $m\in\mathbb{N}$, define:
\[S_n^m = \sum\limits_{w\in A^*, |w|\le m} \left(\frac{1}{|A| + 1}\right)^{|w|} \vphi(w\beta_n), \qquad S^m = \sum\limits_{w\in A^*, |w|\le m} \left(\frac{1}{|A| + 1}\right)^{|w|} \vphi(w\beta). \]
By continuity of $\vphi$, we have for every $m\in\mathbb{N}$ that:
\begin{align*}
\lim\limits_{n\to\infty} S_n^m &=  \sum\limits_{w\in A^*, |w|\le m} \lim\limits_{n\to\infty}\left(\frac{1}{|A| + 1}\right)^{|w|} \vphi(w\beta_n)\\
&=\sum\limits_{w\in A^*, |w|\le m} \left(\frac{1}{|A| + 1}\right)^{|w|} \vphi(w\beta) = S^m
\end{align*}
(the sum is finite, which means that we can interchange it with the limit). On the other hand, we can bound the difference between $\psi(\beta_n)$ and $S_n^m$ as follows:
\begin{align*}
|\psi(\beta_n) - S_n^m| &\le \sum\limits_{w\in A^*, |w|> m} \left(\frac{1}{|A| + 1}\right)^{|w|}\cdot |\vphi(w\beta_n)| \\
&\le \sum\limits_{l > m} \left(\frac{|A|}{|A| + 1}\right)^{l} \cdot W \\
&=  (|A| +1)W \cdot\left(\frac{|A|}{|A| + 1}\right)^{m+1} .
\end{align*}
Likewise, we have $|\psi(\beta) - S^m| \le (|A| +1)W \cdot\left(\frac{|A|}{|A| + 1}\right)^{m+1}$. Thus, for every $m$ we obtain:
\begin{align*}
\limsup\limits_{n\to\infty} |\psi(\beta_n) - \psi(\beta)| &\le \limsup\limits_{n\to\infty}\Big[ |\psi(\beta_n) - S_n^m| + |S_n^m- S^m| + |S^m -  \psi(\beta)| \Big]\\
&\le 2(|A| +1)W \cdot\left(\frac{|A|}{|A| + 1}\right)^{m+1}.
\end{align*}
Since $m$ can be arbitrarily large, we obtain $\limsup\limits_{n\to\infty} |\psi(\beta_n) - \psi(\beta)|  = 0$, as required.

\bigskip

\textbf{Why is $\psi$ prefix-monotone?} Take any $\beta, \gamma \in A^\omega$. We have to show that either $\psi(u\beta) \ge \psi(u\gamma)$ for all $u\in A^*$ or $\psi(u\beta) \le \psi(u\gamma)$ for all $u\in A^*$.

Since $\vphi$ is prefix-monotone, then either $\vphi(w\beta) \ge \vphi(w \gamma)$ for all $w\in A^*$ or $\vphi(w\beta) \le \vphi(w \gamma)$ for all $w\in A^*$. Up to swapping $\beta$ and $\gamma$, we may assume that $\vphi(w\beta) \ge \vphi(w \gamma)$ for all $w\in A^*$. Then for any $u\in A^*$ the difference
\[\psi(u\beta) - \psi(u\gamma) = \sum\limits_{w\in A^*} \left(\frac{1}{|A| + 1}\right)^{|w|} \big[\vphi(wu\beta) - \vphi(wu\gamma)\big]\]
consists of non-negative terms. Hence $\psi(u\beta) \ge \psi(u\gamma)$ for all $u\in A^*$, as required.

\bigskip

\textbf{Why is $\psi$ shift-deterministic?} Take any $a\in A$ and $\beta, \gamma\in A^\omega$ with $\psi(\beta) = \psi(\gamma)$. We have to show that $\psi(a\beta) = \psi(a\gamma)$. Indeed, assume that
\[0 = \psi(\beta) - \psi(\gamma) = \sum\limits_{w\in A^*} \left(\frac{1}{|A| + 1}\right)^{|w|} \big[\vphi(w\beta) - \vphi(w\gamma)\big].\]
If this series contains a non-zero term, then it must contain a positive term and a negative term. But this contradicts prefix-monotonicity of $\vphi$. So all the terms in this series must be $0$. That is, we have $\vphi(w\beta) - \vphi(w\gamma) = 0$ for every $w\in A^*$. Therefore,
\[\psi(a\beta) - \psi(a\gamma) = \sum\limits_{w\in A^*} \left(\frac{1}{|A| + 1}\right)^{|w|} \big[\vphi(wa\beta) - \vphi(wa\gamma)\big] = 0.\]

\bigskip

\textbf{Why $\vphi = g\circ \psi$ for some continuous non-decreasing $g\colon \psi(A^\omega)\to\mathbb{R}$?}  Let us first show that
\begin{equation}
\label{for_g}
\mbox{$\vphi(\alpha) > \vphi(\beta) \implies \psi(\alpha) > \psi(\beta)$ for all  $\alpha, \beta \in A^\omega$.}
\end{equation}
Indeed, if $\vphi(\alpha) > \vphi(\beta)$, then we also have $\vphi(w\alpha) \ge \vphi(w\beta)$ for every $w\in A^*$, by prefix-monotonicity of $\vphi$. Now, by definition,
\[\psi(\alpha) - \psi(\beta) = \sum\limits_{w\in A^*} \left(\frac{1}{|A| + 1}\right)^{|w|} \big[\vphi(w\alpha) - \vphi(w\beta)\big].\]
All the terms in this series are non-negative, and the term corresponding to the empty $w$ is strictly positive. So we have $\psi(\alpha) > \psi(\beta)$, as required.

Now, let us demonstrate that \eqref{for_g} implies that $\vphi = g\circ \psi$ for some non-decreasing $g\colon \psi(A^\omega) \to\mathbb{R}$. Namely, define $g$ as follows. For $x\in \psi(A^\omega)$, take an arbitrary $\gamma \in \psi^{-1}(x)$ and set $g(x) = \vphi(\gamma)$. First, why do we have $\vphi = g \circ \psi$? By definition, $g(\psi(\alpha)) = \vphi(\gamma)$ for some $\gamma\in A^\omega$ with $\psi(\alpha) = \psi(\gamma)$. By \eqref{for_g} we also have $\vphi(\alpha) = \vphi(\beta)$, so $g(\psi(\alpha)) = \vphi(\gamma) = \vphi(\alpha)$, as required. Now, why is $g$ non-decreasing? I.e., why for all $x, y\in  \psi(A^\omega)$ we have $x \le y \implies g(x) \le g(y)$? Indeed, $g(x) = \vphi(\gamma_x), g(y) = \vphi(\gamma_y)$ for some $\gamma_x \in \psi^{-1}(x)$ and $\gamma_y\in \psi^{-1}(y)$. Now, since $x \le y$, we have $x = \psi(\gamma_x) \le \psi(\gamma_y) = y$. By taking the contraposition to \eqref{for_g}, we get that $g(x) = \vphi(\gamma_x) \le \vphi(\gamma_y) = g(y)$, as required.

Finally, we show that any $g\colon \psi(A^\omega)\to\mathbb{R}$ such that $\vphi = g\circ \psi$ must be continuous. For that, we show that $|g(x) - g(y)| \le |x - y|$ for all $x, y\in\psi(A^\omega)$.  Take any $\alpha, \beta \in A^\omega$ with $x = \psi(\alpha)$ and $y = \psi(\beta)$. By prefix-monotonicity of $\vphi$ we have that either $\vphi(w\alpha) \ge \vphi(w\beta)$ for all $w\in A^*$ or $\vphi(w\alpha) \le \vphi(w\beta)$ for all $w\in A^*$. Up to swapping $x$ and $y$, we may assume that the first option holds. Then 
\[\psi(\alpha) - \psi(\beta) = \sum\limits_{w\in A^*} \left(\frac{1}{|A| + 1}\right)^{|w|} \big[\vphi(w\alpha) - \vphi(w\beta)\big] \ge \vphi(\alpha) - \vphi(\beta) \ge 0.\]
On the left here we have $x - y$, and on the right we have $\vphi(\alpha) - \vphi(\beta) = g\circ \psi(\alpha) - g\circ \psi(\beta) = g(x) - g(y)$. That is, we get $x - y \ge g(x) - g(y) \ge 0$, and it implies that $|g(x) - g(y)| \le |x - y|$.
\end{proof}

\section{Fixed point argument}
\label{sec:fixed}
Here we present a way of establishing positional determinacy of continuous prefix-monotone shift-deterministic payoffs (Proposition \ref{cpmsd}) via a fixed point argument. Together with Proposition \ref{phi_to_psi}, this constitutes our second proof of Theorem \ref{thm:char}.

Obviously, for any shift-deterministic payoff $\vphi\colon A^\omega\to\mathbb{R}$ and for any $a\in A$ there is a unique function $\sh[a,\vphi]\colon \vphi(A^\omega)\to\vphi(A^\omega)$ such that 
$\sh[a,\vphi]\big(\vphi(\beta)\big) = \vphi(a\beta)$ for all $\beta\in A^\omega$.
\begin{rem}
Sometimes, when $\vphi$ is clear from the context, we will simply write $\sh[a]$ instead of $\sh[a,\vphi]$.
\end{rem}
\begin{clm}
\label{pm_eq}
A shift-deterministic payoff $\vphi\colon A^\omega\to\mathbb{R}$ is prefix-monotone if and only if $\sh[a,\vphi]$ is non-decreasing for every $a\in A$.
\end{clm}
\begin{proof} 

A statement that $\sh[a,\vphi]$ is non-decreasing for every $a\in A$ is equivalent to the condition \textbf{\emph{(a)}} of Proposition \ref{fairly_mixing}. In turn, by Claim \ref{obv}, this condition is equivalent to a statement that $\vphi$ is prefix-monotone and shift-deterministic.
\end{proof}

We use this notation to introduce so-called \emph{Bellman's equations}, playing a key role in our fixed point argument.
\begin{defi}
\label{bellman}
Let $A$ be a finite set, $\vphi\colon A^\omega\to\mathbb{R}$ be a shift-deterministic payoff and $G = \langle V, V_\Max, V_\Min, E\rangle$ be an $A$-labeled game graph.

The following equations in $\x\in \vphi(A^\omega)^V$ are called \textbf{Bellman's equations} for $\vphi$ in $G$:
\begin{align}
\label{max_opt}
\x_u &= \max\limits_{e \in E, \source(e) = u} \sh[\lab(e),\vphi](\x_{\target(e)}), \qquad \mbox{for } u\in V_\Max,\\
\label{min_opt}
\x_u &= \min\limits_{e \in E, \source(e) = u} \sh[\lab(e), \vphi](\x_{\target(e)}), \qquad \mbox{for } u\in V_\Min.
\end{align}
\end{defi}

The most important step of our argument is to show the existence of a solution to Bellman's equations.
\begin{prop}
\label{exists_solution}
For any finite set $A$, for any continuous prefix-monotone shift-deterministic payoff $\vphi\colon A^\omega\to\mathbb{R}$ and for any $A$-labeled game graph $G$ there exists a solution to Bellman's equations for $\vphi$ in $G$.
\end{prop}

This proposition requires some additional work. We first discuss why does it imply that all continuous prefix-monotone shift-deterministic payoffs are positionally determined. Assume that we are given a solution $\x$ to  (\ref{max_opt}--\ref{min_opt}). How can one extract an equilibrium of positional strategies from it? For that, we take any pair of positional strategies that use only  \emph{$\x$-tight} edges. Here an edge $e$ is called  $\x$-tight if $\x_{\source(e)} = \sh[a,\vphi](\x_{\target(e)})$. Note that each node must contain an out-going $\x$-tight edge (this will be any edge on which the maximum/minimum in (\ref{max_opt}--\ref{min_opt}) is attained for this node). So clearly each player has at least one positional strategy which only uses $\x$-tight edges. It remains to show that for continuous prefix-monotone shift-deterministic $\vphi$, any two such strategies of the players form an equilibrium.
\begin{lem}
\label{from_bell_to_pos}
If $A$ is a finite set, $\vphi\colon A^\omega\to\mathbb{R}$ is a continuous prefix-monotone shift-deterministic payoff, and $\x\in \vphi(A^\omega)^V$ is a solution to (\ref{max_opt}--\ref{min_opt}) for an $A$-labeled game graph $G = \langle V, V_\Max, V_\Min, E\rangle$, then the following holds. Let $\sigma^*$ be a positional strategy of Max and $\tau^*$ be a positional strategy of Min such that $\sigma^*(V_\Max)$ and $\tau^*(V_\Min)$ consist only of $\x$-tight edges. Then $(\sigma^*,\tau^*)$ is an equilibrium in $G$.
\end{lem}
\begin{proof}
For brevity, we will omit $\vphi$ in the notation $\sh[a,\vphi]$. 
We will also use a notation \[\sh[a_1 a_2 \ldots a_n] = \sh[a_1] \circ \sh[a_2]\circ\ldots \ldots \circ \sh[a_n]\] for $n\in\mathbb{N}, a_1 a_2 \ldots a_n \in A^n$. In particular, $\sh[\mbox{empty string}]$ will denote the identity function.

It is enough to show that
\begin{itemize}
\item \textbf{\emph{(a)}} for any $v\in V$ and for any $\mathcal{P} \in \Cons(v,\sigma^*)$  we have 
\[\vphi\circ\lab\big(\mathcal{P}\big) \ge \x_v^*.\]
\item \textbf{\emph{(b)}} for any $v\in V$ and for any $\mathcal{P} \in \Cons(v,\tau^*)$  we have 
\[\vphi\circ\lab\big(\mathcal{P}\big) \le \x_v^*.\]
\end{itemize}
Indeed, from these two inequalities we obtain that $\Val[\sigma^*](v) \ge \x_v^* \ge \Val[\tau^*](v)$ for every $v\in V$. But on the other hand, $\Val[\sigma^*](v) \le \vphi\circ\lab\big(\mathcal{P}^{\sigma^*,\tau^*}_v\big) \le \Val[\tau^*](v)$. That is, we get that $\Val[\sigma^*](v) = \Val[\tau^*](v)$ for every $v\in V$, and this by definition means  that $(\sigma^*, \tau^*)$ is an equilibrium.

We only show the first item, the second one can be proved similarly.
 Let $e_n$ be the $n$th edge of $\mathcal{P}$ for $n\ge 1$. Define $v_n = \target(e_1 e_2\ldots e_n)\in V$ and
$T_n = \sh[\lab(e_1 \ldots e_n)](\x_{v_n}^*)\in\vphi(A^\omega)$.
We also set $v_0 = v$ and $T_0 = \x_v^*$. Note that due to the continuity of $\vphi$ we have that $\lim_{n\to\infty} T_n = \vphi \circ\lab(\mathcal{P})$. Indeed, $\x_{v_n}^*\in \vphi(A^\omega)$, so there exists $\beta_n\in A^\omega$ with $\x_{v_n}^* = \vphi(\beta_n)$. Hence, $\vphi(\lab(e_1 \ldots e_n) \beta_n) = \sh[\lab(e_1 \ldots e_n)](\vphi(\beta_n)) = \sh[\lab(e_1 \ldots e_n)](\x_{v_n}^*)  = T_n$. On the other hand, the first $n$ letters of $\lab(e_1 \ldots e_n) \beta_n$  and  $\lab(\mathcal{P})$ coincide. Hence, $\vphi(\lab(e_1 \ldots e_n) \beta_n) = T_n$ converges to  $\vphi \circ\lab(\mathcal{P})$ as $n\to\infty$, as required.

So \textbf{\emph{(a)}} is equivalent to a statement that $\lim_{n\to\infty} T_n \ge T_0$. To show this statement, we demonstrate  that $T_{n+1} \ge T_n$ for every $n$.  Indeed, assume first that $v_n \in V_\Max$. Then, since $\mathcal{P}$ is consistent with $\sigma^*$, we have $e_{n + 1} = \sigma^*(v_n)$. In particular, $e_{n + 1}$ is $\x$-tight, by the conditions of the lemma. This gives us that $\sh[\lab(e_{n +1})](\x_{v_{n + 1}}^*) = \x_{v_n}^*$. After applying the function $\sh[\lab(e_1 e_2 \ldots e_n)]$ to this equality, we obtain $T_{n + 1} = T_{n}$.

Now, if $v_n\in V_\Min$, then $\sh[\lab(e_{n +1})](\x_{v_{n + 1}}^*)\ge \x_{v_n}^*$ by \eqref{min_opt}. The function $\sh[\lab(e_1 e_2 \ldots e_n)]$ is composed of non-decreasing functions due to Claim \ref{pm_eq}. Hence, after applying this function to the left-hand and the right-hand sides of the inequality $\sh[\lab(e_{n +1})](\x_{v_{n + 1}}^*)\ge \x_{v_n}^*$, we obtain $T_{n+1} \ge T_{n}$.\end{proof}

We now proceed to details of our proof of Proposition \ref{exists_solution}. Consider a function $T\colon \vphi(A^\omega)^V\to\vphi(A^\omega)^V$, mapping $\x \in \vphi(A^\omega)^V$ to the vector of the right-hand sides of (\ref{max_opt}--\ref{min_opt}). We should argue that $T$ has a fixed point. For that, we will construct a continuous metric $D\colon \vphi(A^\omega)^V\times \vphi(A^\omega)^V \to [0, +\infty)$ with respect to which $T$ is \emph{contracting}. 
More precisely, $D(T\x, T\y)$ will always be smaller than $D(\x, \y)$ as long as $\x$ and $\y$ are distinct. Due to the compactness of the domain of $T$, this will prove that $T$ has a fixed point.

Now, to construct such $D$, we show that for continuous shift-deterministic $\vphi$ there must be a continuous metric $d\colon \vphi(A^\omega)\times \vphi(A^\omega) \to [0, +\infty)$ such that all functions $\sh[a, \vphi], a \in A$ are $d$-contracting. Once we have such $d$,  we let $D(\x, \y)$ be the maximum of $d(\x_a, \y_a)$ over $a\in V$. Checking that $T$ is contracting with respect to such $D$ will be rather straightforward.
The main technical challenge is to prove the existence of $d$. We do so  
via the following general fact about compositions of continuous functions.

\begin{thm}
\label{general_fact}
Let $K\subseteq \mathbb{R}$ be a compact set, $m\ge 1$ be a natural number and $f_1, \ldots, f_m\colon K \to K$ be $m$ continuous functions. Then the following two conditions are equivalent:
\begin{itemize}
\item \textbf{(a)} for every $a_1 a_2 a_3\ldots \in\{1, 2, \ldots, m\}^\omega$ we have
$\lim_{n\to\infty} \diam\big(f_{a_1}\circ f_{a_2} \circ \ldots \circ f_{a_n}(K)\big) = 0$
(by $\diam(S)$ for $S\subseteq \mathbb{R}$ we mean $\sup_{x,y\in S} |x - y|$);
\item \textbf{(b)} there exists a continuous metric $d\colon K\times K\to[0,+\infty)$ such that $f_1, f_2, \ldots, f_m$ are all $d$-contracting (a function $h\colon K \to K$ is called $d$-contracting if for all $x,y\in K$ with $x\neq y$ we have $d(h(x), h(y)) < d(x, y)$).
\end{itemize}
If $f_1, \ldots, f_m$ are non-decreasing, then these two conditions are equivalent to the following condition:
\begin{itemize} \item \textbf{(c)} there exists a continuous metric $d\colon K\times K\to[0,+\infty)$ such that, first, $f_1, f_2, \ldots, f_m$ are all $d$-contracting, and second, for all $x, y, s, t \in K$ we have $x\le s\le t\le y\implies d(s, t) \le d(x, y)$.  
\end{itemize}
\end{thm}
We postpone the proof of this result to the end of this section.

To derive Proposition \ref{exists_solution} from this theorem, we first show that it is applicable to functions $\sh[a,\vphi], a\in A$ for continuous shift-deterministic $\vphi$.

\begin{prop}
\label{check}
Let  $A$ be a finite set and $\vphi\colon A^\omega\to\mathbb{R}$ be a continuous shift-deterministic payoff. Then the functions $\sh[a,\vphi], a\in A$ are continuous and satisfy the condition \textbf{(a)} of Theorem \ref{general_fact} for $K = \vphi(A^\omega)$.
\end{prop}
\begin{proof}
We use the same abbreviations with respect to the notation $\sh[a,\vphi]$ as in the proof of Lemma \ref{from_bell_to_pos}.

Let us first demonstrate that $\sh[a]$ is continuous for every $a\in A$.
Consider any $x\in \vphi(A^\omega)$ and any infinite sequence $\{x_n\in\vphi(A^\omega)\}_{n\in\mathbb{N}}$ such that $\lim_{n\to\infty} x_n = x$. We shall show that $\lim_{n\to\infty} \sh[a](x_n) = \sh[a](x)$. It is enough to show that $\sh[a](x)$ is the only limit point of the sequence $\{\sh[a](x_n)\}_{n\in\mathbb{N}}$. In other words, w.l.o.g~we may assume that the limit $\lim_{n\to\infty} \sh[a](x_n)$ exists, and our goal is to show that it equals $\sh[a](x)$.

 Let $\beta_n\in A^\omega$ be such that $x_n = \vphi(\beta_n)$. Due to the compactness of $A^\omega$, there exists $\beta\in A^\omega$ such that  any open set $\mathcal{S} \subseteq A^\omega$, containing $\beta$, also contains $\beta_n$ for infinitely many $n$. Indeed, otherwise any point of $A^\omega$ is contained in an open set which covers only finitely many elements of the sequence $\{\beta_n\}_{n\in\mathbb{N}}$. A collection of such open sets would be an open cover of $A^\omega$ without a finite subcover (no finite subcover can have $\beta_n$ for all $n$).

 For every $k\in\mathbb{N}$ there exists $n_k \ge k$ such that the first $k$ letters of $\beta_{n_k}$ and $\beta$ coincide. Indeed, consider a word $u\in A^k$, consisting of the first $k$ letters of $\beta$. An open set $\mathcal{S} = u A^\omega$ contains $\beta$. Hence, there are infinitely many $n$ such that $\beta_n\in u A^\omega$, or, equivalently, such that $\beta_n$ starts with $u$. In particular, there exists such $n$ which is at least as large as $k$.
 
 Due to continuity of $\vphi$, we have  that $\lim_{k\to\infty} \vphi(\beta_{n_k})= \vphi(\beta)$. On the other hand, $ \lim_{k\to\infty} \vphi(\beta_{n_k}) =\lim_{k\to\infty} x_{n_k} = x$. Hence, $\vphi(\beta) = x$. Using the continuity of $\vphi$ again, we get 
\begin{align*}
\lim_{n\to\infty} \sh[a](x_n) &= \lim_{n\to\infty} \sh[a](\vphi(\beta_n)) =  \lim_{n\to\infty} \vphi(a\beta_n) =  \lim_{k\to\infty} \vphi(a\beta_{n_k}) \\
&= \vphi(a\beta) && \text{by continuity of $\vphi$}\\
&= \sh[a](\vphi(\beta)) = \sh[a](x),&& \text{because $\vphi(\beta) = x$}
\end{align*}
as required.

Now, let us show that $\sh[a], a\in A$ satisfy item \textbf{\emph{(a)}} of Theorem \ref{general_fact} for $K = \vphi(A^\omega)$. By definition of $\sh[a]$, we have that
\[\sh[a_1\ldots a_n]\big(\vphi(A^\omega)\big) = \vphi(a_1a_2 a_3\ldots a_n A^\omega)\]
for every $n\in\mathbb{N}$ and $a_1 a_2\ldots a_n\in A^n$.
Thus, it is enough to establish that 
\[\lim\limits_{n\to\infty} \diam\big(\vphi(a_1 a_2 \ldots a_n A^\omega)\big) = 0\]
for any $a_1 a_2 a_3\ldots \in A^\omega$. This is a simple consequence of the continuity of $\vphi$. Indeed, assume for contradiction that for some  $a_1 a_2 a_3\ldots \in A^\omega$ we have  $\diam\big(\vphi(a_1 a_2 \ldots a_n A^\omega)\big) > \varepsilon$ for infinitely many $n$. Then for infinitely many $n$ there exist $\beta_n, \gamma_n \in a_1 a_2 \ldots a_n A^\omega$ with $|\vphi(\beta_n) - \vphi(\gamma_n)| \ge \varepsilon$. At the same time, by continuity of $\vphi$, both $\vphi(\beta_n)$ and $\vphi(\gamma_n)$ must converge to $\vphi(a_1 a_2 a_3\ldots)$, contradiction. \qedhere

\end{proof}

We finally derive Proposition \ref{exists_solution} from Theorem \ref{general_fact} and Proposition \ref{check}.
This will finish our second proof of the fact that all continuous prefix-monotone payoffs are positionally determined.
\begin{proof}[Proof of Proposition \ref{exists_solution}]
We use the same abbreviations with respect to the notation  $\sh[a,\vphi]$ as in the proof of Lemma \ref{from_bell_to_pos}.

Define a mapping $T\colon K^V \to K^V$, where $K = \vphi(A^\omega)$, as follows:
\begin{align}
\label{T1}
T(\x)_u &=\max\limits_{e \in E, \source(e) = u} \sh[\lab(e)](\x_{\target(e)}), \qquad \mbox{for } u\in V_\Max,\\
\label{T2}
T(\x)_u &= \min\limits_{e \in E, \source(e) = u} \sh[\lab(e)](\x_{\target(e)}), \qquad \mbox{for } u\in V_\Min.
\end{align}
Recall that $K$ is a compact set (because $A^\omega$ is compact and $\vphi$ is continuous).
It is enough to show that $T$ has a fixed point. By Proposition \ref{check}, the functions $\sh[a], a\in A$ are continuous (which means that $T$ is also continuous) and satisfy the item  \textbf{\emph{(a)}} of Theorem \ref{general_fact}. By Claim \ref{pm_eq}, the functions $\sh[a], a\in A$ are non-decreasing. Hence, these functions satisfy the item \textbf{\emph{(c)}} of Theorem \ref{general_fact}. That is, there exists a continuous metric $d\colon K\times K\to[0,+\infty)$ such that, first,  the function $\sh[a]$ is $d$-contracting for every $a\in A$, and second,  for every $x, s, t, y\in K$ we have $x\le s \le t \le y\implies d(s, t) \le d(x, y)$.

Define a metric $D\colon K^V\times K^V\to [0, +\infty)$ as follows:
\[D(\mathbf{x}, \mathbf{y}) = \max\limits_{u\in V} d(\mathbf{x}_u, \mathbf{y}_u).\]
It is enough to show $D(T(\mathbf{x}), T(\mathbf{y})) < D(\mathbf{x}, \mathbf{y})$ for all $\x, \y\in K^V, \x\neq \y$. Indeed, assume that this inequality is already established. Consider a point $\x^* \in K^V$ minimizing $D(\x, T(\x))$. Such $\x^*$ exists because $D(\x, T(\x))$ is continuous and $K^V\times K^V$ is a compact set. If $\x^* \neq T(\x^*)$, then $D(T(\x^*), T\circ T(\x^*)) < D(\x^*, T(\x^*))$, contradiction.

Now, take any $\x, \y\in K^V, \x\neq \y$. Let $u\in V$ be such that $D(T(\mathbf{x}), T(\mathbf{y})) = d(T(\x)_u, T(\y)_u)$. Assume w.l.o.g.~that $u\in V_\Max$. Also, up to swapping $\x$ and $\y$, we  may assume that $T(\x)_u \le T(\y)_u$. Let $e$ be an edge on which the maximum in \eqref{T1} is attained for $\y$. That is, the source of $e$ is $u$ and $T(\y)_u = \sh[\lab(e)](\y_{w})$, where $w = \target(e)$. On the other hands, by \eqref{T1} applied to $\x$, we get $\sh[\lab(e)](\x_w)\le T(\x)_u$. Overall, 
\[\sh[\lab(e)](\x_w) \le T(\x)_u \le T(\y)_u = \sh[\lab(e)](\y_w).\]

Since for any  $x, s, t, y\in K$ it holds that $x\le s \le t \le y\implies d(s, t) \le d(x, y)$, we get:
\[d(T(\x)_u, T(\y)_u) \le d\big(\sh[\lab(e)](\x_w), \sh[\lab(e)](\y_w)\big).\]
If $\x_w = \y_w$, then $0 = d(T(\x)_u, T(\y)_u) = D(T(\x), T(\y)) < D(\x, \y)$, because $\x\neq \y$. Now, if $\x_w \neq \y_w$, then due to the fact that the function $\sh[\lab(e)]$ is $d$-contracting, we have:
\begin{align*}
d\big(\sh[\lab(e)](\x_w), \sh[\lab(e)](\y_w)\big) < d(\x_w, \y_w) \le D(\x, \y),
\end{align*}
which gives us $D(T(\x), T(\y)) = d(T(\x)_u, T(\y)_u) \le d\big(\sh[\lab(e)](\x_w), \sh[\lab(e)](\y_w)\big) < D(\x, \y)$, as required.
\end{proof}

We finish this section with the missing proof of Theorem \ref{general_fact}.
\begin{proof}[Proof of Theorem \ref{general_fact}]

For the sake of readability, we will use the following notation. First, we will denote $f_i$ by $f[i]$. Moreover, we will abbreviate
\[f[a_1 a_2 \ldots a_n] = f[a_1]\circ f[a_2]\circ\ldots \circ f[a_n]\]
for $n\in\mathbb{N}, a_1 a_2\ldots a_n\in \{1, 2, \ldots, m\}^n$. In particular, $f[\mbox{empty word}]$ will denote the identity function).

\begin{lem}
\label{konig}
The condition \textbf{(a)} of Theorem \ref{general_fact} is equivalent to the following condition: for every $\varepsilon > 0$ there are only finitely many $w\in\{1, 2, \ldots, m\}^*$ such that 
\[\diam\big(f[w](K)\big) > \varepsilon.\]
\end{lem}
\begin{proof}
Assume that the condition \textbf{\emph{(a)}} of Theorem \ref{general_fact} holds. Take any $\varepsilon > 0$. Call  $w\in\{1, 2, \ldots, m\}^*$ \emph{bad} if $\diam\big(f[w](K)\big) > \varepsilon$. We have to show that the number of bad $w$ if finite. Assume for contradiction that the number of bad $w$ is infinite. Observe that any prefix of a bad $w$ is also bad. Indeed, if $w = uv$ for some  $u,v\in\{1, 2, \ldots, m\}^*$, then $f[w](K) = f[u] \circ f[v](K) \subseteq f[u](K)$. Hence, by K\H{o}nig's Lemma, there exists $\alpha = a_1 a_2 a_3\ldots\in\{1, 2, \ldots, m\}^\omega$ such any finite prefix of $\alpha$ is bad. Observe that 
\[\liminf\limits_{n\to\infty}\diam\big(f[a_1 a_2\ldots a_n](K)\big) \ge \varepsilon.\]
This is a contradiction with the condition \textbf{\emph{(a)}} of Theorem \ref{general_fact}.

The opposite direction of the lemma is obvious.
\end{proof}

The rest of the proof is organized as follows. We first show that \textbf{\emph{(a)}$\implies$\emph{(b)}}. Then we observe that the same proof establishes \textbf{\emph{(a)} $\implies$ \emph{(c)}} when $f[1], \ldots, f[m]$ are non-decreasing. Finally, we show that \textbf{\emph{(b)} $\implies$ \emph{(a)}}. Since, obviously, \textbf{\emph{(c)} $\implies$ \emph{(b)}}, this will establish Theorem \ref{general_fact}.

\bigskip

\textbf{Proof of \emph{(a)} $\implies$ \emph{(b)}.}
Define
\begin{equation}
\label{our_metric}
d\colon K\times K\to[0,+\infty), \qquad d(x, y) = \sup\limits_{w\in \{1, \ldots, m\}^*} \big(2 - 2^{-|w|}\big) \cdot \big|f[w](x) - f[w](y)\big|.
\end{equation}
First, we obviously have $d(x, x) = 0$ and $d(x, y) = d(y, x)$. Notice also that $d(x, y) \ge |x - y|$, so $d(x, y) > 0$ for $x\neq y$. In turn, the triangle inequality $d(x, y) \le d(x, z) + d(z, y)$ holds because, first, it holds for every $w\in \{1, 2, \ldots, m\}^*$ in \eqref{our_metric}, and second, the supremum of the sums is at most the sum of the supremums.  These considerations show that $d$ is a metric.

Note that the supremum in \eqref{our_metric} is always attained on some $w\in \{1, 2, \ldots, m\}^*$. Indeed, if $d(x, y) = 0$, then it is attained already on the empty word. Assume now that $d(x, y) > 0$.  By Lemma \ref{konig} all but finitely many $w\in \{1, 2, \ldots, m\}^*$ satisfy
$\diam\big(f[w](K)\big) \le d(x, y)/3$.
 So  only finitely many terms in \eqref{our_metric} are bigger than $2d(x, y)/3$, and hence the supremum (which is $d(x, y) > 2d(x,y)/3$) must be attained on one of them.

This already implies that $f[i]$ is $d$-contracting for every $i\in \{1, 2, \ldots, m\}$. Indeed, take any $x, y \in K$. Then for some $w\in \{1, 2, \ldots, m\}^*$ we have:
\[d\big(f[i](x), f[i](y)\big) = (2 - 2^{-|w|}) \cdot \left|f[w]\big(f[i](x)) - f[w]\big(f[i](y)\big)\right|.\]
We have to show that if $x\neq y$, then $d\big(f[i](x), f[i](y)\big) < d(x, y)$. If $d\big(f[i](x), f[i](y)\big) = 0$, there is nothing to prove. Otherwise, the quantity
\[\left|f[w]\big(f[i](x)) - f[w]\big(f[i](y)\big)\right|\]
is positive. Therefore, we can write:
\begin{align*}
d\big(f[i](x), f[i](y)\big) &= (2 - 2^{-|w|}) \cdot \left|f[w]\big(f[i](x)) - f[w]\big(f[i](y)\big)\right|\\
&< (2 - 2^{-|w| - 1}) \cdot \left|f[w]\big(f[i](x)) - f[w]\big(f[i](y)\big)\right| \\
&=  (2 - 2^{-|wi|}) \cdot \left|f[wi](x) - f[wi](y)\right|
\le d(x, y).
\end{align*}

It remains to show that $d$ is continuous. Consider any $(x_0, y_0)\in K\times K$. We have to show that for any $\varepsilon > 0$ there exists $\delta > 0$ such that for all $(x, y) \in K\times K$ with $|x - x_0| + |y - y_0| \le \delta$ we have $|d(x, y) - d(x_0, y_0)|\le \varepsilon$.

By Lemma \ref{konig}, there exists $n \in\mathbb{N}$ such that for all $w\in \{1, 2, \ldots, m\}^*$ with $|w| \ge n$ we have:
\[\diam\big(f[w](K)\big) \le \varepsilon/6.\]
In particular, this means that all terms in \eqref{our_metric} corresponding to $w\in \{1, 2, \ldots, m\}^*$ with $|w| \ge n$ are at most $\varepsilon/3$. Hence, for every $(x, y)\in K\times K$ we have that $d(x, y)$ is $(\varepsilon/3)$-close to $d_n(x, y)$, where
\[d_n(x, y) = \max\limits_{w\in A^*, |w| < n} \big(2 - 2^{-|w|}\big) \cdot \big|f[w](x) - f[w](y)\big|.\]
Now, notice that the function $d_n$ is continuous (as a composition of finitely many continuous functions). Hence there exists $\delta > 0$ such that for all $(x, y)\in K\times K$ with $|x - x_0| + |y - y_0| \le \delta$ we have $|d_n(x, y) - d_n(x_0, y_0)| \le \varepsilon/3$. Obviously, for all such $(x, y)$ we also have $|d(x, y) - d(x_0, y_0)| \le \varepsilon$.

\begin{rem}
Here we observe that if $f[1], \ldots, f[m]$ are non-decreasing, then this construction establishes \textbf{\emph{(a)} $\implies$ \emph{(c)}}. That is, we show that if $f[1], \ldots, f[m]$ are non-decreasing, then $d(s, t) \le d(x, y)$ for all $x, s, t, y \in K$ with $x\le s\le t\le y$. Indeed, in this case the function $f[w]$ is non-decreasing for every $w\in\{1, 2, \ldots, m\}^*$. Hence we have $f[w](x) \le f[w](s) \le f[w](t) \le f[w](y)$ and $|f[w](s) - f[w](t)| \le |f[w](x) - f[w](y)|$. By \eqref{our_metric}, this gives us $d(s,t) \le d(x, y)$.
\end{rem}

\bigskip

\textbf{Proof of \emph{(b)} $\implies$ \emph{(a)}}. We show that for every $\varepsilon > 0$  there exists $n\in\mathbb{N}$ such that for all $w\in\{1, 2, \ldots, m\}^*$ with $|w|\ge n$ it holds that
\[\diam\big(f[w](K)\big) \le \varepsilon.\]
Obviously, this implies \textbf{\emph{(a)}}.

Define $T = \{(x, y) \in K\times K \mid |x - y| \ge \varepsilon\}$. Note that $T$ is a compact set. A function $d(x, y)/|x - y|$ is continuous on $T$. Hence, there exists 
\[z = \min\limits_{(x, y)\in T} d(x, y)/|x - y|.\]
Observe that $z > 0$. Indeed, for some $(x, y) \in T$ we have $z = d(x, y)/|x - y|$.  By definition of $T$, we have $|x - y| \ge \varepsilon$. Hence $x\neq y$ and $d(x, y)$ is positive, as well as $z$.

Now, define $S = \{(x, y) \in K\times K \mid d(x, y) \ge z\cdot \varepsilon\}$. Again, $S$ is a compact set. Consider a function:
\[h(x, y) = \max\limits_{i\in \{1, \ldots, m\}} \frac{d\big(f[i](x), f[i](y)\big)}{d(x, y)}.\]
The function $h$ is continuous on $S$ (we never have $0$ in its denominator on $S$). Hence there exists
\[\lambda = \max_{(x, y) \in S} h(x, y).\]
The function $h$ is non-negative, so $\lambda \ge 0$. Let us show that $\lambda < 1$. Indeed, for some $(x, y) \in S$ we have $\lambda = h(x, y)$. By definition of $h$, for some $i\in \{1, 2,\ldots, m\}$ we have:
\[\lambda = \frac{d\big(f[i](x), f[i](y)\big)}{d(x, y)}.\]
Since $(x, y) \in S$, we have $d(x, y) \ge z\cdot \varepsilon > 0$. Hence $x\neq y$. Now,  $f[i]$ is $d$-contracting. Therefore $d\big(f[i](x), f[i](y)\big) < d(x, y)$ and $\lambda < 1$.

Define $D = \sup_{x,y\in K} d(x, y)$. If $D = 0$, then $K$ consists of a singe point, which means that the condition \textbf{\emph{(a)}} trivially holds. From now on we assume that $D > 0$. 
Take any $n\in\mathbb{N}$ such that 
\[\lambda^n < \frac{z\varepsilon}{D}.\]
We claim that for any $w\in\{1, 2\ldots, m\}^*$ with $|w| \ge n$ we have $\diam\big( f[w](K) \big) \le \varepsilon$. We only have to show this for $w$ of length exatly $n$. This is because if $w^\prime$ is of length at least $n$, then $f[w^\prime](K)$ is contained in $f[w](K)$, where $w$ is a prefix of $w^\prime$ of length $n$.

So take any $w = a_1 a_2 \ldots a_n\in\{1, 2\ldots, m\}^n$.
Let us first establish that:
\begin{equation}
\label{d_diam}
\sup\limits_{x,y\in K} d\big(f[w](x), f[w](y)\big) \le z\varepsilon.
\end{equation}
Assume for contradiction that $d\big(f[w](x), f[w](y)\big) > z\varepsilon$ for some $x,y\in K$. Define $w_{\ge i} = a_i a_{i + 1}\ldots a_n$ for $i = 1, \ldots, n$, and let $w_{\ge n + 1}$ be the empty string. Set 
\[F_i = d\big(f[w_{\ge i}](x), f[w_{\ge i}](y)\big).\] 
Note that $F_1 = d\big(f[w](x), f[w](y)\big)$ and $F_{n + 1} = d(x,y)$.
Since $f[w_\ge i] = f[a_i]\circ f[w_{\ge i + 1}]$, and since $f[a_i]$ is $d$-contracting, for every $i = 1, \ldots, n$ we have
\[F_i = d\big(f[w_{\ge i}](x), f[w_{\ge i}](y)\big) \le d\big(f[w_{\ge i + 1}](x), f[w_{\ge i + 1}](y)\big) = F_{i + 1}.\]
 In fact, if $F_{i + 1} \ge z\varepsilon$, then, by definition of $\lambda$, it holds that $F_i \le \lambda F_{i + 1}$. Recall that $F_1 = d\big(f[w](x), f[w](y)\big) > z\varepsilon$.
Since $F_1 \le F_2 \le \ldots \le F_{n + 1}$, we have $F_i \ge z\varepsilon$ for every $i$. Therefore, $F_1 \le \lambda F_2 \le\ldots \le \lambda^n F_{n + 1}$. On the other hand, $F_{n + 1} = d(x, y) \le D$. Hence, $z\varepsilon < F_1 \le \lambda^n D$. But by definition of $n$ we have $\lambda^n D < z\varepsilon$, contradiction.

It remains to show that \eqref{d_diam} implies that $\diam\big(f[w](K)\big) \le \varepsilon$. We do so by showing for every $x,y\in K$ that $|f[w](x) - f[w](y)| > \varepsilon \implies d\big(f[w](x), f[w](y) \big) > z\varepsilon$. This means that if $\diam\big(f[w](K)\big) > \varepsilon$, then \eqref{d_diam} cannot hold.

Now, if $|f[w](x) - f[w](y)| > \varepsilon$, then $(f[w](x), f[w](y)) \in T$, so
\[\frac{d\big(f[w](x), f[w](y)\big)}{|f[w](x) - f[w](y)|} \ge \min\limits_{(x, y)\in T} \frac{d(x, y)}{|x - y|} = z\neq 0.\]
Therefore, $d\big(f[w](x), f[w](y)\big) \ge z\cdot |f[w](x) - f[w](y)| > z\varepsilon$. \qedhere

\end{proof}

\section{The Structure of Continuous Positional Payoffs}
\label{sec:fixed_app}

In this section, we give an explicit description of the set of continuous positionally determined payoffs, see Theorem \ref{exhaustive_method} below. Then, in Proposition \ref{disc_char}, we use our description to give an alternative definition of the class of multi-discounted payoffs. Finally, in Proposition \ref{not_multi}, we give an example of a continuous positionally determined payoff which, in a quite strong sense, does not ``reduce'' to multi-discounted payoffs.

 We start with some terminology. Let $K\subseteq \mathbb{R}$ be a compact set. We call a family of $m$ continuous functions $f_1, \ldots, f_m\colon K \to K$ a \emph{contracting base} if there exists a continuous metric $d\colon K\times K \to [0, +\infty)$ such that $f_i$ is $d$-contracting for every $i\in\{1, \ldots, m\}$ (in other words, if $f_1, \ldots, f_m$ satisfy the item \textbf{\emph{(b)}} of Theorem \ref{general_fact}). If $f_1, \ldots, f_m$ are non-decreasing, then we call this family of functions a non-decreasing contracting base.

\begin{clm} 
\label{contracting_base}
Let $K\subseteq \mathbb{R}$ be a compact set and $f_1, \ldots, f_m\colon K\to K$ be $m$ continuous functions forming a contracting base. Then for any $a_1 a_2 a_3\ldots\in \{1, 2, \ldots, m\}^\omega$ we have
\[\left|\bigcap\limits_{n = 1}^\infty f_{a_1} \circ f_{a_2}\circ \ldots \circ f_{a_n}(K)\right| = 1.\]
Moreover, for any $x\in K$, the quantity $f_{a_1} \circ f_{a_2}\circ \ldots \circ f_{a_n}(x)$ converges to the unique element of this intersection as $n\to\infty$.
\end{clm}
\begin{proof}
This intersection is non-empty 
due to Cantor's intersection theorem. To show that this intersection contains just one point, observe that $f_1, \ldots, f_m$ satisfy the item \textbf{\emph{(b)}} of Theorem \ref{general_fact} by definition. Hence, they also satisfy the item  \textbf{\emph{(a)}} of this theorem. This means that the diameter of this intersection is $0$. 

As for the second claim, note that the distance between the unique element of our intersection and the point $f_{a_1} \circ f_{a_2}\circ \ldots \circ f_{a_n}(x)$ is at most $\diam f_{a_1} \circ f_{a_2}\circ \ldots \circ f_{a_n}(K)$, because both these points belong to $f_{a_1} \circ f_{a_2}\circ \ldots \circ f_{a_n}(K)$. It remains to refer to the item  \textbf{\emph{(a)}} of Theorem \ref{general_fact} once again.
\end{proof}
This claim means that any contracting base  $f_1, \ldots, f_m\colon K\to K$ induces a payoff $\psi[f_1, \ldots, f_m]\colon\{1, \ldots, m\}^\omega\to K \subseteq\mathbb{R}$, defined by
\[\{\psi[f_1, \ldots, f_m](a_1 a_2 a_3\ldots)\} = \bigcap\limits_{n = 1}^\infty f_{a_1} \circ f_{a_2}\circ \ldots \circ f_{a_n}(K),\]
for every $a_1 a_2 a_3\ldots \in \{1, 2,  \ldots, m\}^\omega$. By the second part of Claim \ref{contracting_base}, we have:
\[\psi[f_1, \ldots, f_m](a_1 a_2 a_3\ldots) = \lim\limits_{n\to\infty} f_{a_1}\circ\ldots\circ f_{a_n}(x)\]
for every $x\in K$.
\begin{clm} 
\label{contracting_payoff}
Let $K\subseteq\mathbb{R}$ be a compact set and $f_1, \ldots, f_m\colon K\to K$ be $m$ continuous functions forming a contracting base. Then the payoff $\psi = \psi[f_1, \ldots, f_m]$, induced by $f_1, \ldots, f_m$, is continuous and shift-deterministic. Moreover, for all $i\in\{1, 2, \ldots, m\}$ and $\alpha = a_1a_2 a_3\ldots\in\{1, 2, \ldots, m\}^\omega$ we have:
\begin{equation}
\label{psi_shift_eq}
\psi(i \alpha) = f_i(\psi(\alpha)).
\end{equation}
\end{clm}
\begin{proof}
Let us first establish \eqref{psi_shift_eq}. Take any $x\in K$.
 By Claim \ref{contracting_base}, we have
$\psi(i\alpha) = \lim_{n\to\infty} f_i \circ f_{a_1}\circ\ldots\circ f_{a_n}(x)$. By continuity of $f_i$, we get $\psi(i\alpha) = f_i\big(\lim_{n\to\infty} f_{a_1}\circ\ldots\circ f_{a_n}(x)\big)= f_i(\psi(\alpha))$.

This immediately implies that $\psi$ is shift-deterministic. To show that $\psi$ is continuous, we use Proposition \ref{cont_help}. Take any $\alpha = a_1 a_2 a_3\ldots \in\{1, 2, \ldots, m\}^\omega$ and any infinite sequence $\{\beta_n\}_{n\ge 1}$ of elements of $\{1, 2, \ldots, m\}^\omega$ such that $\alpha$ and $\beta_n$ have the same prefixes of length $n$, for every $n\ge 1$. We have to show that $\psi(\beta_n)\to\psi(\alpha)$ as $n\to\infty$. By \eqref{psi_shift_eq}, both $\psi(\beta_n)$ and $\psi(\alpha)$ belong to the set $f_{a_1}\circ f_{a_2}\circ\ldots \circ f_{a_n}(K)$. Hence, the difference between $\psi(\beta_n)$ and $\psi(\alpha)$ does not exceed the diameter of this set. But by the item \textbf{\emph{(a)}} of Theorem \ref{general_fact}, the diameter of this set converges to $0$ as $n\to\infty$.
\end{proof}

\begin{thm}
\label{exhaustive_method}
Let $m \ge 1$ be a natural number. Set\footnote{We assume that edge labels are natural numbers, for the simplicity of notation.} $A = \{1, 2, \ldots, m\}$. Then the set of continuous positionally determined payoffs from $A^\omega$ to $\mathbb{R}$ coincides with the set of $\vphi\colon A^\omega\to \mathbb{R}$ that can be obtained in the following 5 steps.
\begin{itemize}
\item \textbf{Step 1.} Take a compact set $K\subseteq \mathbb{R}$.
\item \textbf{Step 2.} Take a continuous metric $d\colon K\times K\to [0,+\infty)$.
\item \textbf{Step 3.} Take $m$ continuous non-decreasing $d$-contracting functions $f_1, f_2, \ldots, f_m\colon K\to K$. They will form a non-decreasing contracting base.
\item \textbf{Step 4.} Consider the payoff $\psi = \psi[f_1, \ldots, f_m]$ induced by $f_1, \ldots, f_m$.
\item \textbf{Step 5.} Choose a continuous non-decreasing function $g\colon \psi(A^\omega)\to\mathbb{R}$ and set $\vphi = g\circ \psi$.
\end{itemize}  
\end{thm}
\begin{proof}
Assume first that $\vphi\colon A^\omega\to\mathbb{R}$ is continuous and positionally determined. Then $\vphi$ is prefix-monotone by Theorem \ref{thm:char}. By Proposition \ref{phi_to_psi}, there is a continuous prefix-monotone shift-deterministic payoff $\psi\colon A^\omega\to\mathbb{R}$  and a continuous non-decreasing function $g\colon \psi(A^\omega)\to\mathbb{R}$ such that $\vphi = g\circ \psi$.  Set $K = \psi(A^\omega)$. Note that $K$ is compact due to the continuity of $\psi$. Define $f_i = \sh[i,\psi]$.  By Claim \ref{pm_eq}, the functions $f_1, f_2, \ldots, f_m$ are non-decreasing. By Proposition \ref{check}, the functions $f_1, f_2, \ldots, f_m$ are continuous and satisfy the item \textbf{\emph{(a)}} of Theorem \ref{general_fact}. Hence, they form a non-decreasing contracting base with respect to some continuous metric $d\colon K\times K \to [0,+\infty)$. It remains to show that $\psi$ coincides with the payoff induced by $f_1, \ldots, f_m$. For that, take any $x\in K = \psi(A^\omega)$. By the second part of Claim \eqref{contracting_base}, it is sufficient to show that $\psi(a_1 a_2 a_3\ldots) = \lim_{n\to\infty} f_{a_1}\circ f_{a_2}\circ \ldots \circ f_{a_n}(x)$ for every $a_1 a_2 a_3\ldots\in A^\omega$.  There exists $\beta\in A^\omega$ such that $x = \psi(\beta)$. Observe that $f_{a_1} \circ f_{a_2}\circ \ldots \circ f_{a_n}(x) = \sh[a_1, \psi]\circ\sh[a_2, \psi]\circ\ldots \circ \sh[a_n,\psi](\psi(\beta)) = \psi(a_1 a_2\ldots a_n\beta)$. This quantity converges  to $\psi(a_1 a_2 a_3\ldots)$ as $n\to\infty$ due to the continuity of $\psi$.

In turn, assume that $\vphi$ was obtained in these 5 steps. By Theorem \ref{thm:char}, we only have to show that $\vphi$ is continuous and prefix-monotone. First, by Claim \ref{contracting_payoff}, we have that $\psi$ is continuous. Since $\vphi = g\circ \psi$ and $g$ is continuous, we have that $\vphi$ is continuous as well. In turn, since $f_1, \ldots, f_m$ are non-decreasing, from \eqref{psi_shift_eq} we get that $\psi$ is prefix-monotone. This easily implies that $\vphi$ is also prefix-monotone. Indeed, if $\vphi$ is not prefix-monotone, then $g\circ\psi(u\beta) > g\circ \psi(u\gamma)$ and $g\circ\psi(v\beta) > g\circ \psi(v\gamma)$ for some $u, v\in A^*$ and $\beta,\gamma\in A^\omega$. Since $g$ is non-decreasing, this implies that $\psi(u\beta) > \psi(u\gamma)$ and $\psi(v\beta) < \psi(v\gamma)$, contradiction with the prefix-monotonicity of $\psi$.

\end{proof}
\begin{rem}
Recall that we did not use the continuity of $g$ from Proposition \ref{phi_to_psi} in the inductive argument, but we use it in the proof of Theorem \ref{exhaustive_method}.
\end{rem}

Next, we characterize the class of multi-discounted payoffs, using the language of Theorem \ref{exhaustive_method}.

\begin{prop}
\label{disc_char}
Let $m \ge 1$ be a natural number and set $A = \{1, 2,\ldots, m\}$. Then the set of multi-discounted payoffs from $A^\omega$ to $\mathbb{R}$ coincides with the set of $\psi\colon A^\omega\to\mathbb{R}$ that can be obtained as in Theorem \ref{exhaustive_method} with the following additional requirements:
\begin{itemize}
\item $K = [-W, W]$ for some $W > 0$;
\item $d$ is a standard metric $d(x, y) = |x - y|$;
\item $f_1, \ldots, f_m$ are affine functions with the slope from $[0,1)$. That is, for each $a\in\{1, \ldots, m\}$ there exists $\lambda(a) \in [0, 1)$ and $w(a)\in\mathbb{R}$ such that $f_a(x) = \lambda(a) x + w(a)$ (observe that we must have $\lambda(a) W + w(a) \le W$ and $\lambda(a) (-W) + w(a) \ge -W$, because $f_a\colon [-W, W]\to[-W, W]$).
\item $g$ is the identity function.
\end{itemize}
\end{prop}
\begin{proof}
First, assume that $\psi\colon A^\omega\to\mathbb{R}$ is a multi-discounted payoff. Take $\lambda\colon A\to[0, 1)$ and $w\colon A\to\mathbb{R}$ such that
\[\psi(a_1 a_2 a_3\ldots) = \sum\limits_{n = 1}^\infty \lambda(a_1) \cdot \ldots \cdot \lambda(a_{n-1}) \cdot w(a_n)\]
for all $a_1 a_2 a_3\ldots \in A^\omega$. Obviously, there exists $W > 0$ such that $\lambda(a) W + w(a) \le W$ and $\lambda(a) (-W) + w(a) \ge -W$ for all $a\in A$. Set $K = [-W, W]$ and $f_a(x) = \lambda(a) x + w(a)$ for every $a\in A$. Obviously, $f_1, \ldots, f_m\colon[-W,W]\to[-W,W]$ form a non-decreasing contracting base with respect to a standard metric $d(x, y) = |x - y|$. It remains to show that $\psi$ coincides with the payoff induced by $f_1, \ldots, f_m$. Indeed,
\begin{align*}
\psi(a_1 a_2 a_3\ldots) &= \sum\limits_{n = 1}^\infty \lambda(a_1) \cdot \ldots \cdot \lambda(a_{n-1}) \cdot w(a_n) \\
&=\lim\limits_{n\to\infty}\big[ w(a_1) + \lambda(a_1) \cdot w(a_2) + \ldots + \lambda(a_1)\cdot\ldots \cdot\lambda(a_{n - 1}) w(a_n)\big]\\
&= \lim\limits_{n\to\infty}\big[w(a_1) + \lambda(a_1) \big(w(a_2) + \ldots + \big(w(a_{n - 1}) + \lambda(a_{n - 1}) w(a_n) \big)\big)\big]\\
&= \lim\limits_{n\to\infty}\big[ f_{a_1}\circ f_{a_2}\circ\ldots f_{a_n}(0)\big].
\end{align*}
This computation also establishes the opposite direction of the proposition. Indeed, assume that $\psi$ was obtained as in Theorem \ref{exhaustive_method} with the requirements of our proposition. Then $\psi$ must be equal to the payoff induced by $f_1, \ldots, f_m$, where $f_1, \ldots, f_m$ are affine functions with the slope from $[0,1)$. By writing down the same chain of equalities as above, we get that $\psi$ is multi-discounted.
\end{proof}

We also construct a continuous positionally determined payoff which does not ``reduce'' to the multi-discounted ones, in a sense of the following definition. 

\begin{defi}
Let $A$ be a finite set, $\vphi, \psi\colon A^\omega\to\mathbb{R}$ be two payoffs, and $G$ be an $A$-labeled game graph. We say that $\vphi$ \textbf{positionally reduces} to $\psi$ \textbf{inside} $G$ if any pair of positional strategies in $G$ which is an equilibrium for $\psi$ is also an equilibrium for $\vphi$.
\end{defi}
This definition has an algorithmic motivation. Namely, note that finding a positional equilibrium for $\psi$ in $G$ is at least as hard as for $\vphi$, provided that $\vphi$ reduces to $\psi$ inside $G$. There are classical reductions from Parity to Mean Payoff games~\cite{jurdzinski1998deciding} and from Mean Payoff to Discounted games~\cite{zwick1996complexity} that work in exactly this way. See also~\cite{gimbert2008applying} for a reduction from \emph{Priority} Mean Payoff games to Multi-Discounted games. As far as we know, our next proposition provides the first example of a positionally determined payoff which does not reduce to multi-discounted payoffs in this sense.

\begin{prop}
\label{not_multi}
There exist a finite set $A$, a continuous positionally determined payoff $\vphi\colon A^\omega\to\mathbb{R}$ and an $A$-labeled game graph $G$ such that there exists no multi-discounted payoff to which $\vphi$ reduces inside $G$.
\end{prop}
\begin{proof}
It is sufficient to establish the following lemma.
\begin{lem}
\label{not_multi_strong}
There exist a finite set $A$, a continuous positionally determined payoff $\vphi\colon A^\omega\to\mathbb{R}$ and three pairs $(\alpha_1, \beta_1), (\alpha_2, \beta_2), (\alpha_3, \beta_3)\in A^\omega\times A^\omega$ of ultimately periodic infinite words such that:
\begin{itemize}
\item for all $i\in\{1,2,3\}$ we have $\vphi(\alpha_i) > \vphi(\beta_i)$
\item for every multi-discounted payoff $\psi\colon A^\omega\to\mathbb{R}$ there exists $i\in\{1,2,3\}$ such that $\psi(\alpha_i) \le \psi(\beta_i)$.
\end{itemize}
\end{lem}

Indeed, assume that this lemma is proved. Consider a game graph from Figure \ref{3l}, consisting of three pairs of ``lassos''.  The only optimal positional strategy of Max there w.r.t.~$\vphi$ is to go to the left from $v_1$, $v_2$ and $v_3$. On the other hand, any multi-discounted payoff has an optimal positional strategy which for some $i\in\{1, 2, 3\}$ goes to the right from $v_i$. Hence, there is no multi-discounted payoff to which $\vphi$ positionally reduces insides the game graph from Figure \ref{3l}.

\begin{figure}[h!]
\centering
\begin{tikzpicture}[gamegraph]
\node[draw,circle] (v1) {$v_1$};
\node[draw,circle,left=1.5cm of v1] (v1l) {};
\node[draw,circle,right=1.5cm of v1] (v1r) {};
\draw[transition] (v1)--(v1l);
\draw[transition] (v1)--(v1r);
\draw[transition] (v1l) to [out=120,in=-120, distance=2cm] (v1l);
\draw[transition] (v1r) to [out=60,in=-60, distance=2cm] (v1r);
\node[draw=none, above right = 0.2 and 0.2cm of v1l] (a1) { $\alpha_1$};
\node[draw=none, above left = 0.2 and 0.2cm of v1r] (b1) { $\beta_1$};

\node[draw,circle, below=1cm of v1] (v2) {$v_2$};
\node[draw,circle,left=1.5cm of v2] (v2l) {};
\node[draw,circle,right=1.5cm of v2] (v2r) {};
\draw[transition] (v2)--(v2l);
\draw[transition] (v2)--(v2r);
\draw[transition] (v2l) to [out=120,in=-120, distance=2cm] (v2l);
\draw[transition] (v2r) to [out=60,in=-60, distance=2cm] (v2r);
\node[draw=none, above right = 0.2 and 0.2cm of v2l] (a2) { $\alpha_2$};
\node[draw=none, above left = 0.2 and 0.2cm of v2r] (b2) { $\beta_2$};

\node[draw,circle, below=1cm of v2] (v3) {$v_3$};
\node[draw,circle,left=1.5cm of v3] (v3l) {};
\node[draw,circle,right=1.5cm of v3] (v3r) {};
\draw[transition] (v3)--(v3l);
\draw[transition] (v3)--(v3r);
\draw[transition] (v3l) to [out=120,in=-120, distance=2cm] (v3l);
\draw[transition] (v3r) to [out=60,in=-60, distance=2cm] (v3r);
\node[draw=none, above right = 0.2 and 0.2cm of v3l] (a3) { $\alpha_3$};
\node[draw=none, above left = 0.2 and 0.2cm of v3r] (b3) { $\beta_3$};

\end{tikzpicture}
  \caption{All nodes are owned by Max. For every $i = 1, 2, 3$, the node $v_i$ has two lassos $\mathcal{L}_i$ and $\mathcal{R}_i$ starting at it, one going to the left, and the other going to the right. We label their edges in such a way that $\lab(\mathcal{L}_i) = \alpha_i$ and $\lab(\mathcal{R}_i) = \beta_i$. This is possible because $\alpha_1, \alpha_2, \alpha_3$ and $\beta_1, \beta_2, \beta_3$ are ultimately periodic.}
\label{3l}
\end{figure}
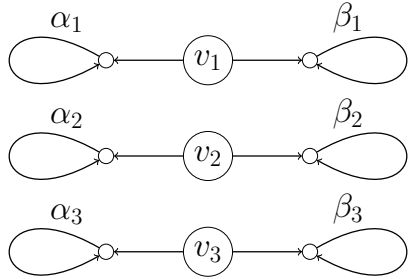

To show Lemma \ref{not_multi_strong}, we observe that following property of the multi-discounted payoffs.

\begin{clm}
\label{aaa}
Let $A$ be a finite set and $\psi\colon A^\omega\to\mathbb{R}$ be a multi-discounted payoff. Then there are no $a, b\in A, \gamma \in A^\omega$ such that
\begin{align*}
\psi(a\gamma) &> \psi(b\gamma),\\
\psi(aa\gamma) &< \psi(bb\gamma),\\
\psi(aaa\gamma) &> \psi(bbb\gamma).
\end{align*}
\end{clm}
\begin{proof}
Assume for contradiction that such $a, b, \gamma$ exist. Take $\lambda\colon A\to [0,1)$ and $w\colon A\to [0,1)$ defining $\psi$ as in \eqref{multi_disc}.
Set $\lambda = \lambda(a)$, $\mu = \lambda(b)$, $u = w(a), v = w(b)$ and $x = \psi(\gamma)$. Then $\lambda, \mu\in [0, 1)$ and
\begin{align}
\label{one}
\lambda x + u &> \mu x + v,\\
\label{two}
\lambda^2 x + (1 + \lambda)u &< \mu^2 x + (1 + \mu) v,\\
\label{three}
\lambda^3 x + (1 + \lambda + \lambda^2)u &> \mu^3 x + (1 + \mu + \mu^2) v.
\end{align}
 Multiply \eqref{one} by $\lambda + \mu + \lambda\mu$, multiply \eqref{two} by $-(1 +\lambda + \mu)$, multiply \eqref{three} by $1$ and take the sum. This will give us $0 > 0$, contradiction.
\end{proof}

To finish a proof of Lemma \ref{not_multi_strong}, we construct a continuous positionally determined payoff $\vphi\colon\{1, 2, 3\}^\omega\to\mathbb{R}$ such that:
\begin{align*}
\vphi(1 3^\omega) &> \vphi(2 3^\omega), \\
\vphi(11 3^\omega) &< \vphi(22 3^\omega), \\
\vphi(111 3^\omega) &> \vphi(222 3^\omega).
\end{align*}
For that, we use Theorem \ref{exhaustive_method}.
Namely, we set $K = [0, 1]$ and $d(x, y) = |x - y|$. Next, we let $f_1 = \frac{x}{2}$, $f_3 = \frac{x}{2} + \frac{1}{2}$. These two functions are clearly $d$-contracting. Finally, we let $f_2\colon [0,1]\to[0,1]$ be a piece-wise linear function whose graph has the following break-points:
\[(0, 0),\,\, (0.26, 0.11),\,\,  (0.49, 0.26),\,\,  (1, 0.49).\]
Observe that its slope is always from $[0, 1)$, so $f_2$ is also $d$-contracting. So $f_1, f_2, f_3$ is a non-decreasing contracting base. Let $\vphi\colon\{1,2,3\}^\omega\to\mathbb{R}$ be the payoff induced by $f_1, f_2, f_3$, that is, 
\[\vphi(a_1 a_2 a_3\ldots) = \lim\limits_{n\to\infty} f_{a_1} \circ f_{a_2}\circ\ldots \circ f_{a_n}(1).\]
(Of course, 1 can here can be changed to any point from $[0,1]$, but 1 is the most convenient for computations below.)
By Theorem \ref{exhaustive_method}, we have that $\vphi$ is a continuous positionally determined payoff. Now,
it is easy to see that $\vphi(3^\omega) = 1$ and
\begin{align*}
\vphi(1 3^\omega) = 0.5 &> \vphi(2 3^\omega) = 0.49, \\
\vphi(11 3^\omega) = 0.25 &< \vphi(22 3^\omega) = 0.26, \\
\vphi(111 3^\omega) = 0.125 &> \vphi(222 3^\omega) = 0.11. \tag*{\qedhere}
\end{align*}
\end{proof}

\section{Strategy improvement argument}
\label{sec:str}
Here we establish the existence of a solution to Bellman's equations (Proposition \ref{exists_solution}) via the \emph{strategy improvement}. This will yield our third proof of Theorem \ref{thm:char}. We start with an observation that the vector of values of a positional strategy always gives a solution to a restriction of Bellman's equations to edges that are consistent with this strategy. 

\begin{lem}
\label{cons_lemma}
Let $A$ be a finite set, $\vphi\colon A^\omega\to\mathbb{R}$ be a continuous prefix-monotone shift-deterministic payoff and $G = \langle V, V_\Max, V_\Min, E\rangle$ be an $A$-labeled game graph. Then for every positional strategy $\sigma$ of Max in $G$ we have for all $u\in V$:
\begin{align*}
\Val[\sigma](u) &=  \min\limits_{e \in E^\sigma, \source(e) = u} \sh[\lab(e), \vphi]\big(\Val[\sigma](\target(e))\big) \,\, \mbox{for } u\in V_\Min.
\end{align*}
(If $u\in V_\Max$, the minimum is over a single edge $e = \sigma(u)$. If $u\in V_\Min$, the minimum is over all edges that start at $u$).
\end{lem}
\begin{rem}
Technically, Bellman's equations are over $\x\in\vphi(A^\omega)^V$. So we have to argue that $\Val[\sigma](u)\in\vphi(A^\omega)$ for every $u\in V$. This is because $\Val[\sigma](u)$ is the infimum of some subset of $\vphi(A^\omega)$. In turn, since $\vphi$ is continuous, we have that $\vphi(A^\omega)$ is compact, and hence is closed.
\end{rem}
\begin{proof}[Proof of Lemma \ref{cons_lemma}]
For brevity, we will denote $C_u = \Cons(u)$.
By definition, $\Val[\sigma](u)$ is the infimum of the image of $\vphi\circ\lab$ on the set $C_u$. Now, the set $C_u$ is exactly the set of infinite paths that start at $u$ and consist only of edges from $E^\sigma$. So we can write:
\[C_u = \bigcup\limits_{\substack{e\in E^\sigma \\ \source(e) = u}}eC_{\target(e)}.\]
The infimum of a union of finitely many sets is the minimum of the infimums of these sets. So we get:
\[\Val[\sigma](u) =  \min\limits_{\substack{e\in E^\sigma \\ \source(e) = u}} \inf\vphi \circ\lab\left(eC_{\target(e)}\right).\]
It is sufficient to show that:
 \begin{equation}
\label{eC}
\inf\vphi \circ\lab\left(eC_{\target(e)}\right) = \sh[\lab(e)]\big(\Val[\sigma](\target(e))\big).
\end{equation}

For any $a\in A,\mathcal{S}\subseteq A^\omega$, by definition of $\sh[a]$, we can write:
\[\vphi\big(a\mathcal{S}\big) = \sh[a]\big(\vphi(\mathcal{S})\big).\]
After applying this to $a = \lab(e), \mathcal{S} = \lab\big(C_{\target(e)}\big)$, we obtain:
\[\vphi \circ\lab\left(eC(\target(e)\right) = \sh[\lab(e), \vphi] \big(\vphi\circ\lab\left(C_{\target(e)}\right) \big).\]
Now, since $\sh[\lab(e)]$ is non-decreasing (by Claim \ref{pm_eq}) and continuous (by Proposition \ref{check}), we can interchange $\inf$ and $\sh[\lab(e)]$. This gives us:
\begin{align*}
\inf\sh[\lab(e)] &\big(\vphi\circ\lab\left(C_{\target(e)}\right) \big)\\ &= 
\sh[\lab(e)] \big(\inf\vphi\circ\lab\left(C_{\target(e)}\right)\big) \\
&= \sh[\lab(e)]\big(\Val[\sigma](\target(e))\big).
\end{align*}
Hence, \eqref{eC} is proved.
\end{proof}
Next, take a positional strategy $\sigma$ of Max.
If the vector $\{\Val[\sigma](u)\}_{u\in V}$ happens to be a solution to the Bellman's equations, then we are done. Otherwise by Lemma \ref{cons_lemma} there must exist an edge $e\in E$ with $\source(e) \in V_\Max$ such that $\Val[\sigma](\source(e)) < \sh[\lab(e), \vphi]\big(\Val[\sigma](\target(e))\big)$. We call edges satisfying this property \emph{$\sigma$-violating}. We show that \emph{switching} $\sigma$ to any $\sigma$-violating edge gives us a positional strategy which \emph{improves} $\sigma$.
\begin{lem}
\label{switching_lemma}
Let $A$ be a finite set, $\vphi\colon A^\omega\to\mathbb{R}$ be a continuous prefix-monotone shift-deterministic payoff and $G = \langle V, V_\Max, V_\Min, E\rangle$ be an $A$-labeled game graph. Next, let $\sigma$ be a positional strategy  of Max in $G$. Assume that the vector $\Val[\sigma] = \{\Val[\sigma](u)\}_{u\in V}$ does not satisfy (\ref{max_opt}--\ref{min_opt}) and let $e^\prime \in E$ be any $\sigma$-violating edge. Define a positional strategy $\sigma^\prime$ of Max as follows:
\[\sigma^\prime(u) = \begin{cases}e^\prime & u = \source(e^\prime), \\ \sigma(u) & \mbox{otherwise}.\end{cases}\]
Then $\sum\limits_{u\in V} \Val[\sigma^\prime](u) > \sum\limits_{u\in V} \Val[\sigma](u)$.
\end{lem}
\begin{proof}
For $\x\in \vphi(A^\omega)^V$, let the \emph{modified cost} of an edge $e\in E$ with respect to $\x$ be the following quantity:
\[R^\x(e) = \sh[\lab(e)](\x_{\target(e)}) - \x_{\source(e)}.\]

We need the following ``potential transformation lemma''  (its analog for discounted payoffs is well-known, see, e.g., \cite[Lemma 3.6]{hansen2013strategy}).

\begin{lem} 
\label{pot_lemma}
Take any $\x \in \vphi(A^\omega)^V$.
Let $\mathcal{P} = e_1 e_2 e_2\ldots$ be an infinite path in $G$. Then there exists an infinite sequence of non-negative real numbers $\lambda_1, \lambda_2, \lambda_3,\ldots$ such that $\lambda_1 = 1$ and 
\[\vphi\circ\lab\big(\mathcal{P}\big) - \x_{\source(\mathcal{P})} = \sum\limits_{n = 1}^\infty \lambda_n \cdot R^\x(e_i).\]
\end{lem}
\begin{proof}
For $u\in V$, let $\beta_u\in A^\omega$ be such that $\x_u = \vphi(\beta_u)$. Define $s_n = \target(e_1 e_2\ldots e_n)$ for $n\ge 1$ and $s_0 = \source(\mathcal{P})$. By the continuity of $\vphi$, we have
\[\vphi\circ\lab\big(\mathcal{P}\big) = \lim\limits_{n\to\infty} \vphi\big(\lab(e_1 e_2\ldots e_n)\beta_{s_n}\big) = \lim\limits_{n\to\infty} \sh[\lab(e_1 e_2\ldots e_n)](\x_{s_n}).\]
Hence we obtain
\begin{align*}
\vphi\circ\lab\big(\mathcal{P}\big) - \x_{s_0} &=  \lim\limits_{n\to\infty} \big(\sh[\lab(e_1 e_2\ldots e_n)](\x_{s_n}) - \x_{s_0}\big) \\
&=  \lim\limits_{n\to\infty}  \sum\limits_{k = 1}^n \big(\sh[\lab(e_1 e_2\ldots e_k)](\x_{s_k}) - \sh[\lab(e_1 e_2\ldots e_{k - 1})](\x_{s_{k - 1}}) \big)\\
&= \sum\limits_{n = 1}^\infty\big(\sh[\lab(e_1 e_2\ldots e_n)](\x_{s_n}) - \sh[\lab(e_1 e_2\ldots e_{n - 1})](\x_{s_{n - 1}}) \big).
\end{align*}
We can write each term in this series as:
\begin{align*}
\sh&[\lab(e_1 e_2\ldots e_n)](\x_{s_n}) - \sh[\lab(e_1 e_2\ldots e_{n - 1})](\x_{s_{n - 1}}) \\
&= \sh[\lab(e_1 \ldots e_{n - 1})] \big(\sh[\lab(e_n)](\x_{s_n})\big) -  \sh[\lab(e_1 e_2\ldots e_{n - 1})](\x_{s_{n - 1}})\\
&= \lambda_n \cdot \big(\sh[\lab(e_n)](\x_{s_n}) - \x_{s_{n - 1}}\big),
\end{align*}
for some $\lambda_n\in[0,+\infty)$, because $\sh[\lab(e_1 \ldots e_{n - 1})]$ is non-decreasing (for $n = 1$ we get $\lambda_1 = 1$ because $\sh[\mbox{empty word}]$ is the identity function). It remains to notice that by definition:
\[\sh[\lab(e_n)](\x_{s_n}) - \x_{s_{n - 1}} = R^\x(e_n)\]
(because $\source(e_n) = s_{n - 1}, \target(e_n) = s_n$).
\end{proof}
We apply this lemma to the vector $\g = \{\Val[\sigma](u)\}_{u\in V}$. Note that by Lemma \ref{cons_lemma} we have $R^\g(e) \ge 0$ for every $e\in E^\sigma$. In turn, since $e^\prime$ is $\sigma$-violating, we have $R^\g(e^\prime) > 0$.

Let us at first show that
\[\Val[\sigma^\prime](u) \ge \Val[\sigma](u) = \g_u\]
for every $u\in V$. In other words, we will demonstrate that $\vphi\circ \lab(\mathcal{P}) \ge \g_u$ for any infinite path $\mathcal{P} = e_1 e_2 e_3\ldots \in \Cons(u,\sigma^\prime)$. Indeed, by Lemma \ref{pot_lemma} we can write:
\begin{equation}
\label{series_ineq}
\vphi\circ \lab(\mathcal{P}) - \g_u = \sum\limits_{n = 1}^\infty \lambda_n R^\g(e_n)
\end{equation}
for some $\lambda_n\in [0,+\infty), \lambda_1 = 1$. All edges of $\mathcal{P}$ are from $E^\sigma\cup\{e^\prime\}$. Hence, all terms in this series are non-negative, and so is the left-hand side. 

To establish that $\sum_{u\in V} \Val[\sigma^\prime](u) > \sum_{u\in V} \Val[\sigma](u)$, it is now enough to show that $\Val[\sigma^\prime](u) > \Val[\sigma](u) = \g_u$ for some $u\in V$. We will show this for $u = \source(e^\prime)$. The first edge of any $\mathcal{P}\in \Cons(u,\sigma^\prime)$ is $e^\prime$. So the first term in \eqref{series_ineq} for any such $\mathcal{P}$ equals $R^\g(e^\prime)$. All the other terms, as we have discussed, are non-negative. Hence, $\vphi\circ\lab(\mathcal{P}) \ge R^\g(e^\prime) + \Val[\sigma](u)$ for any $\mathcal{P}\in \Cons(u,\sigma^\prime)$. Since $R^\g(e^\prime)$ is strictly positive, we get that $\Val[\sigma^\prime](u) > \Val[\sigma](u)$.
\end{proof}
By this lemma, a Max's positional strategy $\sigma^*$ maximizing the quantity  $\sum_{u\in V} \Val[\sigma](u)$ (over positional strategies $\sigma$ of Max) gives a solution to (\ref{max_opt}--\ref{min_opt}). Such $\sigma^*$ exists just because there are only finitely many positional strategies of Max. This finishes our strategy improvement proof of Proposition \ref{exists_solution}. Let us note that the same argument can be carried out with positional strategies of Min (via analogues of Lemma \ref{cons_lemma} and Lemma \ref{switching_lemma} for Min).

\section{Subexponential-time Algorithm}
\label{sec:alg}

In this subsection, we discuss implications of our strategy improvement argument to the \emph{strategy synthesis problem}. The strategy synthesis for a positionally determined payoff $\vphi$ is an algorithmic problem of finding an equilibrium (with respect to $\vphi$) of two positional strategies in a given game graph. It is classical that the strategy synthesis for parity, mean and multi-discounted payoffs payoffs admits a randomized algorithm which is subexponential in the number of nodes~\cite{halman2007simple,bjorklund2005combinatorial}. We obtain the same subexponential bound for all continuous positionally determined payoffs. For that, we use a framework of recursively local-global functions due to Bj{\"o}rklund and Vorobyov~\cite{bjorklund2005combinatorial}.

Let us start with an observation that for continuous positionally determined shift-deterministic payoffs, a non-optimal positional strategy can always be improved by changing it in a single node.
\begin{prop}
\label{local_imp}
Let $A$ be a finite set and $\vphi\colon A^\omega\to\mathbb{R}$ be a continuous positionally determined shift-deterministic payoff. Then for any $A$-labeled game graph $G = \langle V, V_\Max, V_\Min, E\rangle$ the following two conditions hold:
\begin{itemize}
\item if $\sigma$ is a positional strategy of Max in $G$ which is not optimal, then in $G$ there exists a Max's positional strategy $\sigma^\prime$ such that $\left|\{u\in V_\Max \mid \sigma(u) \neq \sigma^\prime(u)\} \right| = 1$ and $\sum_{u\in V} \Val[\sigma^\prime](u) > \sum_{u\in V} \Val[\sigma](u)$;
\item if $\tau$ is a positional strategy of Min in $G$ which is not optimal, then in $G$ there exists a Min's positional strategy $\tau^\prime$ such that $\left|\{u\in V_\Min \mid \tau(u) \neq \tau^\prime(u)\} \right| = 1$ and $\sum_{u\in V} \Val[\tau^\prime](u) < \sum_{u\in V} \Val[\tau](u)$.
\end{itemize}
\end{prop}
\begin{proof}
Assume that $\sigma$ is a positional strategy of Max which is not optimal. First, let us show that the vector $\Val[\sigma]$ cannot be a solution to Bellman's equations. Indeed, by Lemma \ref{cons_lemma}, it holds that $\sigma$ uses only edges that are $\Val[\sigma]$-tight. Hence, if $\Val[\sigma]$ were a solution to Bellman's equation, then, by Lemma \ref{from_bell_to_pos}, strategy $\sigma$ would have been optimal.

Since $\Val[\sigma]$ is not a solution to Bellman's equation, we can take $\sigma^\prime$ as in Lemma \ref{switching_lemma}, obtained by switching $\sigma$ to some $\sigma$-violating edge. The argument for positional strategies of Min is similar.
\end{proof}

It is instructive to visualize this proposition by imagining the set of positional strategies of one of the players (say, Max) as a \emph{hypercube}. Namely, in this hypercube there will be as many dimensions as there are nodes of Max. A coordinate, corresponding to a node $u\in V_\Max$, will take values  in the set of edges that start at $u$. Obviously, vertices of such a hypercube are in a one-to-one correspondence with positional strategies of Max. Let us call two vertices \emph{neighbors} of each other if they differ in exactly one coordinate. Now, Proposition \ref{local_imp} means the following: any vertex $\sigma$, maximizing $\sum_{u\in V} \Val[\sigma](u)$ over its neighbors, also maximizes this quantity over the \emph{whole} hypercube.

So the optimization problem of maximizing  $\sum_{u\in V} \Val[\sigma](u)$ (equivalently, finding an optimal positional strategy of Max) has the following remarkable feature: all its \emph{local} maxima are also \emph{global}. For positional strategies of Min the same holds for the minima. Optimization problems with this feature have been studied in numerous works, starting from a classical area of convex optimization.

Observe that in our case this local-global property is \emph{recursive}; i.e., it holds for any restriction to a \emph{subcube} of our hypercube. Indeed, subcubes correspond to subgraphs of our initial game graph, and for any subgraph we still have Proposition \ref{local_imp}. Bj{\"o}rklund and Vorobyov~\cite{bjorklund2005combinatorial} noticed that a similar phenomenon occurs for all classical positionally determined payoffs. In turn, they showed that any optimization problem on a hypercube with this recursive local-global property admits a randomized algorithm which is subexponential in the dimension of a hypercube.  In our case, this yields a randomized algorithm for the strategy synthesis problem which is subexponential in the number of nodes of a game graph.

Still, this only applies to continuous payoffs that are shift-deterministic (as we have Proposition \ref{local_imp} only for shift-deterministic payoffs). One more issue is that we did not specify how our payoffs are represented. We overcome these difficulties in the following result. Its proof is given in Section \ref{sec:subexp}.

\begin{thm}
\label{cont_alg}
Let $A$ be a finite set and $\varphi\colon A^\omega \to\mathbb{R}$ be a continuous positionally determined payoff.
Consider an oracle which for given $u, v, a, b\in A^*$ tells, whether there exists $w\in A^*$ such that $\vphi(w u (v)^\omega) > \vphi(w a (b)^\omega)$. 
There exists a randomized algorithm, which solves the strategy synthesis problem for $\vphi$ with this oracle in expected $e^{O\left(\log m + \sqrt{n\log m}\right)}$ time for game graphs with $n$ nodes and $m$ edges. In particular, every call to the oracle in the algorithm is for $u, v, a, b\in A^*$ that are of length $O(n)$, and the expected number of the calls is $e^{O\left(\log m + \sqrt{n\log m}\right)}$.
\end{thm}

So to deal with the issue of representation we assume a suitable oracle access to $\vphi$. Still, the oracle from Theorem \ref{cont_alg} might look unmotivated. Here it is instructive to recall that all continuous positionally determined $\vphi$ must be prefix-monotone. For prefix-monotone $\vphi$, the formula $\exists w\in A^*\,\, \vphi(w\alpha) > \vphi(w\beta)$ defines a total preorder on $A^\omega$, and our oracle just compares ultimately periodic infinite words according to this preorder. In fact, it is easy to see that the formula $\exists w\in A^*\,\, \vphi(w\alpha) > \vphi(w\beta)$ defines a total preorder on $A^\omega$ if and \emph{only if} $\vphi$ is prefix-monotone. This indicates a fundamental role of this preorder for prefix-monotone $\vphi$ and justifies the use of the corresponding oracle in Theorem \ref{cont_alg}. Let us note that $\big[\exists w\in A^*\,\, \vphi(w\alpha) > \vphi(w\beta)\big] \iff \vphi(\alpha) > \vphi(\beta)$ if $\vphi$ is additionally shift-deterministic.

\section{Proof of Theorem \ref{cont_alg}}
\label{sec:subexp}

First in Subsection \ref{subsec:reducing} it is demonstrated that w.l.o.g. we may assume that $\vphi$ is shift-deterministic (so that we can use Proposition \ref{local_imp}) and that we are given an oracle which simply compares values of $\vphi$ on ultimately periodic infinite words. Then in Subsection \ref{subsec:rlg} we expose a framework of \emph{recursively local-global functions} due to  Bj{\"o}rklund and Vorobyov. Finally, in Subsection \ref{subsec:deriving} we use this framework to show Theorem \ref{cont_alg} in the assumptions of Subsection \ref{subsec:reducing}.
\subsection{Reducing to shift-deterministic payoffs}
\label{subsec:reducing}
 It is sufficient to establish Theorem \ref{cont_alg} with the following assumptions.
\begin{asm}
\label{ass1}
Payoff $\vphi$ is continuous, positionally determined and \textbf{shift-deterministic}.
\end{asm}
\begin{asm}
\label{ass2}
We are given an oracle which for $u, v, a, b\in A^*$ tells, whether $\vphi(u(v)^\omega) > \vphi(a(b)^\omega)$.
\end{asm}

To justify this, it is enough to show the following lemma.
\begin{lem}
\label{phi_to_psi_refined}
Let $A$ be a finite set and let $\vphi\colon A^\omega\to\mathbb{R}$ be a continuous positionally determined payoff. Then there exist a  continuous positionally determined shift-deterministic payoff $\psi\colon A^\omega\to\mathbb{R}$ and a  non-decreasing function $g\colon \psi(A^\omega) \to\mathbb{R}$ such that $\vphi = g\circ\psi$ and
\[\big[\exists w\in A^*\,\, \vphi(w\alpha) > \vphi(w\beta)\big] \iff \psi(\alpha) > \psi(\beta) \qquad \mbox{ for all } \alpha, \beta\in A^\omega.\]
\end{lem}
Indeed, let $\vphi$ be an arbitrary continuous positionally determined payoff and $\psi$ be as in Lemma \ref{phi_to_psi_refined}. By Proposition \ref{conserv_eq}, an equilibrium for $\psi$ is also an equilibrium for $\vphi = g\circ \psi$. So to solve the strategy synthesis for $\vphi$, it is enough to do so for $\psi$. Clearly, $\psi$ satisfies Assumption \ref{ass1}. Finally, note that the oracle from Assumption \ref{ass2} for $\psi$ simply coincides on every input with the oracle we are given for $\vphi$ in Theorem \ref{cont_alg}.

\begin{proof}[Proof of Lemma \ref{phi_to_psi_refined}]
We take $\psi$ as in the proof of Proposition \ref{phi_to_psi}. It is established there that
\begin{itemize}
\item $\psi$ is continuous, prefix-monotone  (hence positionally determined) and shift-deterministic;
\item $\vphi = g\circ \psi$ for some non-decreasing $g\colon\psi(A^\omega) \to\mathbb{R}$.
\end{itemize}
This information is sufficient to show that
\[\big[\exists w\in A^*\,\, \vphi(w\alpha) > \vphi(w\beta)\big] \implies \psi(\alpha) > \psi(\beta).\]
Indeed, if $g\circ \psi(w\alpha) = \vphi(w \alpha)  > \vphi(w\beta) = g\circ \psi(w\beta)$  for some $w\in A^*$, then we also have $\psi(w\alpha) > \psi(w\beta)$, because $g$ is non-decreasing. Due to prefix-monotonicity of $\psi$, we also have $\psi(\alpha) \ge \psi(\beta)$. It remains to demonstrate that $\psi(\alpha) \neq \psi(\beta)$. Indeed, $\psi(\alpha) = \psi(\beta) \implies \psi(w\alpha) = \psi(w\beta)$ because $\psi$ is shift-deterministic.

To demonstrate that 
\[\psi(\alpha) > \psi(\beta) \implies \big[\exists w\in A^*\,\, \vphi(w\alpha) > \vphi(w\beta)\big]\]
we have to recall the construction of $\psi$.
By \eqref{psi_def}, we can write:
\[\psi(\alpha) - \psi(\beta) = \sum\limits_{w\in A^*} \left(\frac{1}{|A| + 1}\right)^{|w|} \big[\vphi(w\alpha) - \vphi(w\beta)\big]\]
If $\vphi(w\alpha) \le \vphi(w\beta)$ for all $w\in A^*$, then clearly $\psi(\alpha) \le \psi(\beta)$. This is exactly the contraposition to the implication that we have to prove.
\end{proof}

\subsection{Recursively local-global functions}
\label{subsec:rlg}

Fix $d\in\mathbb{N}$. A $d$-dimensional \emph{structure} is a collection $\mathcal{S} = \{S_i\}_{i = 1}^d$ of  $d$ non-empty finite sets $S_1, S_2, \ldots, S_d$. \emph{Vertices} of $\mathcal{S}$ are elements of the Cartesian product $\prod_{i = 1}^d S_i$. Two vertices $\sigma = (\sigma_1, \ldots, \sigma_d), \sigma^\prime = (\sigma_1^\prime, \ldots, \sigma_d^\prime) \in \prod_{i = 1}^d S_i$ of a structure $\mathcal{S}$ are called \emph{neighbors} if there is exactly one $i\in\{1, 2, \ldots, d\}$ such that $\sigma_i \neq \sigma_i^\prime$.
A structure $\mathcal{S}^\prime = \{S_i^\prime\}_{i = 1}^d$ is a \emph{substructure} of a structure $\mathcal{S} = \{S_i\}_{i = 1}^d$ if $S_i^\prime\subseteq S_i$ for every $i\in \{1, 2, \ldots, d\}$.



Let $\mathcal{S}$ be a structure and $f$ be a function from the set of vertices of $\mathcal{S}$ to $\mathbb{R}$. A vertex $\sigma$ of $\mathcal{S}$ is called a \emph{local} maximum of $f$ if $f(\sigma) \ge f(\sigma^\prime)$ for every neighbor $\sigma^\prime$ of $\sigma$ in $\mathcal{S}$. A vertex $\sigma$ is called a \emph{global} maximum of $f$ if $f(\sigma) \ge f(\sigma^\prime)$ for every vertex $\sigma^\prime$ of $\mathcal{S}$. The function $f$ is called \emph{local-global} if all its local maxima are global. The function $f$ is called \emph{recursively local-global} if all its restrictions to substructures of $\mathcal{S}$ are local-global.

Given a structure $\mathcal{S}$ and a function $f$ from the set of vertices of $\mathcal{S}$ to $\mathbb{R}$, we are interested in finding a global maximum of $f$. In~\cite{bjorklund2005combinatorial} Bj{\"o}rklund and Vorobyov obtained the following result.

\begin{thmC}[{\cite[Theorem 5.1]{bjorklund2005combinatorial}}]
\label{thm:BV}
Let $\mathcal{S} = \{S_i\}_{i = 1}^d$ be a $d$-dimensional structure and $f\colon \prod_{i = 1}^d S_i \to\mathbb{R}$ be a recursively local-global function. 

Consider an oracle which, given two vertices $\sigma^1$ and $\sigma^2$ of $\mathcal{S}$ that are neighbors of each other, compares $f(\sigma^1)$ and $f(\sigma^2)$. There is a randomized algorithm which find a  global maximum of $f$ with this oracle in expected
\[e^{O\left(\log m + \sqrt{d\log m}\right)} \mbox{ time},\]
where $m = \sum_{i = 1}^d |S_i|$.
\end{thmC}

\subsection{Deriving Theorem \ref{cont_alg} with Assumptions \ref{ass1} and \ref{ass2}} 
\label{subsec:deriving}
Let $G = \langle V, V_\Max, V_\Min, E\rangle$ be an $A$-labeled game graph in which we want to solve the strategy synthesis. We will only show how to find an optimal positional strategy of Max, the argument for Min is similar.

 Let $d = |V_\Max|$ and $V_\Max = \{u_1, u_2, \ldots, u_d\}$. Define $S_i = \{e\in E\mid \source(e) = u_i\}$. Consider a structure $\mathcal{S} = \{S_i\}_{i = 1}^d$. Obviously, we may identify vertices of $\mathcal{S}$ with positional strategies of $\Max$. Define
\[f\colon \prod_{i = 1}^d S_i \to\mathbb{R}, \qquad f(\sigma) = \sum\limits_{u\in V} \Val[\sigma](u).\]

\begin{lem}
Any global maximum of $f$ is an optimal positional strategy of Max.
\end{lem}
\begin{proof}
Let $\sigma$ be a global maximum of $f$ and $\sigma^*$ be any unformly optimal positional strategy of Max. By uniform optimality of $\sigma^*$, we have $\Val[\sigma^*](u) \ge \Val[\sigma](u)$ for every $u\in V$. On the other hand,  $\sigma$  maximizes the sum of the values (over all positional strategies of Max), so we must have $\Val[\sigma^*](u) = \Val[\sigma](u)$ for every $u\in V$. This means that $\sigma$ is also optimal.
\end{proof}
\begin{lem} The function $f$ is recursively local-global.
\end{lem}
\begin{proof}
A fact that $f$ is local-global is a simple consequence of Proposition \ref{local_imp} (note that by Assumption \ref{ass1}, our payoff satisfies the requirements of this proposition). Indeed, a strategy $\sigma$ which is not a global maximum of $f$ cannot be optimal. Then take $\sigma^\prime$ as in Proposition \ref{local_imp}. It is a neighbor of $\sigma$ with $f(\sigma^\prime) > f(\sigma)$, so $\sigma$ cannot be a local maximum  as well.

To show that $f$ is recursively local-global, it is sufficient to note that substructures of $\mathcal{S}$ correspond to subgraphs of $G$, and for these subgraphs we also have Proposition \ref{local_imp}.
\end{proof}

Due to these two lemmas, if we run the algorithm from Theorem \ref{thm:BV}, we get an optimal positional strategy of Max in expected
\[e^{O\left(\log m + \sqrt{d\log m}\right)} \mbox{ time},\]
where
$m = \sum_{i = 1}^d |S_i|$. Note that $d$ does not exceed the number of nodes of $G$ and $m$ does not exceed the number of edges, so Theorem \ref{cont_alg} follows.

Still, the algorithm from Theorem \ref{thm:BV} requires an oracle which, given any two vertices $\sigma^1$ and $\sigma^2$ of $\mathcal{S}$ that are neighbors of each other, compares $f(\sigma^1)$ and $f(\sigma^2)$. In our case, this oracle, given two positional strategies $\sigma^1, \sigma^2$ of Max that differ from each other in exactly one node, compares the sums of their values:
\[\sum\limits_{u\in V} \Val[\sigma^1](u), \qquad \sum\limits_{u\in V} \Val[\sigma^2](u).\]
We have to perform this comparison using the oracle from Assumption \ref{ass2}.

Assume that the node where $\sigma^1$ and $\sigma^2$ differ is $v$.
Let $G^{1}$ (respectively, $G^{2}$) be a game graph obtained from $G$ by deleting all edges that are not consistent with $\sigma^1$ (resp., $\sigma^2$). Next, let $G^{1,2}$ be a game graph of all edges that appear either in $G^{1}$ or in $G^{2}$.

Observe that in $G^{1,2}$, strategies $\sigma^1, \sigma^2$ are the only two positional strategies of Max (indeed, all nodes of Max except $v$ have exactly one out-going edge in $G^{1,2}$, and $v$ has exactly two).
One of these strategies must be optimal in $G^{1,2}$. So either $\Val[\sigma^1](u) \ge \Val[\sigma^2](u)$ for all $u\in V$ or $\Val[\sigma^1](u) \le \Val[\sigma^2](u)$ for all $u\in V$. This means that 
 \begin{align*}
&\sum\limits_{u\in V} \Val[\sigma^1](u) > \sum\limits_{u\in V} \Val[\sigma^2](u) \\ &\iff \exists u\in V\,\, \Val[\sigma^1](u) > \Val[\sigma^2](u).
\end{align*}
So our task reduces to a task of comparing $\Val[\sigma^1](u)$ and $\Val[\sigma^2](u)$ for $u\in V$.

Assume first that our game graph $G$ is \emph{one-player}. This means that for one of the players it holds that all nodes of this player have out-degree 1. In our case, this must be Min, because Max has two distinct positional strategies $\sigma^1$ and $\sigma^2$. In particular, there is exactly one strategy $\tau$ of Min in $G$, and this strategy is positional (even if there are no nodes controlled by Min, we assume that Min has a unique empty strategy $\tau$). Hence, $\Val[\sigma^1](u) = \vphi\circ\lab\big(\mathcal{P}^{\sigma^1,\tau}_u\big)$ and $\Val[\sigma^2](u) = \vphi\circ\lab\big(\mathcal{P}^{\sigma^2,\tau}_u\big)$. It remains to compare the value of $\vphi$ on  $\lab\big(\mathcal{P}^{\sigma^1,\tau}_u\big)$ and on $\lab\big(\mathcal{P}^{\sigma^2,\tau}_u\big)$ using the oracle from Assumption \ref{ass2}. These two infinite words are written over some lassos in $G$, so we can decompose them as $\lab\big(\mathcal{P}^{\sigma^1,\tau}_u\big) = u(v)^\omega$ and $\lab\big(\mathcal{P}^{\sigma^2,\tau}_u\big) = a(b)^\omega$ in polynomial time.

 Theorem \ref{cont_alg} is already proved for one-player game graphs. Hence, at the cost of increasing the expected running time by a factor of $e^{O\left(\log m+ \sqrt{n\log m}\right)}$, we may assume that we also have an oracle which can solve the strategy synthesis for $\vphi$ in one-player game graphs. Then we can find an optimal positional strategy $\tau^1$  of Min in $G^{1}$ and an optimal  positional strategy $\tau^2$  of Min in $G^{2}$. Indeed, in these two graphs all nodes of Max have exactly one out-going edge. Observe that $\tau^1$ is an optimal response to $\sigma^1$ and $\tau^2$ is an optimal response to $\sigma^2$, so we have:
\[\Val[\sigma^1](u) = \vphi\circ\lab\big(\mathcal{P}^{\sigma^1,\tau^1}_u\big), \qquad \Val[\sigma^2](u) = \vphi\circ\lab\big(\mathcal{P}^{\sigma^2,\tau^2}_u\big).\]
It remains to compare the value of $\vphi$ on $\lab\big(\mathcal{P}^{\sigma^1,\tau^1}_u\big)$ and $\lab\big(\mathcal{P}^{\sigma^2,\tau^2}_u\big)$. We can do this via the oracle from Assumption \ref{ass2}.

\section{Multi-discounted Payoffs and MDPs}
\label{sec:mdp}
In this section, we establish the following result.
\begin{thm}
\label{thm:mdp}
Let $A$ be a finite set and $\vphi\colon A^\omega\to\mathbb{R}$ be a continuous payoffs. Then $\vphi$ is positionally determined in MDPs if and only if $\vphi$ is multi-discounted.
\end{thm}

This theorem disproves the following conjecture of Gimbert~\cite{gimbert2007pure}: ``Any payoff function which is positional for the class of non-stochastic one-player
games is positional for the class of Markov decision processes''. Indeed, by Proposition \ref{not_multi}, there exists a continuous positionally determined payoff which is not multi-discounted. By Theorem \ref{sec:mdp}, this payoff is not positionally determined in MDPs.

 A fact that multi-discounted payoffs are positionally determined in MDPs (and in two-player stochastic games as well) is classical~\cite{puterman2014markov}. In the rest of this section, we show that any continuous payoff which is positionally determined in MDPs is multi-discounted.
First, we establish the following two necessary conditions.
\begin{prop}
\label{dist_nec}
Let $A$ be a finite set and $\vphi\colon A^\omega\to\mathbb{R}$ be a continuous payoff which is positionally determined in MDPs. Then there are no $a\in A$, $\beta, \gamma, \delta\in A^\omega$, $(p_1, p_2, p_3), (q_1, q_2, q_3) \in [0, +\infty)^3$ such that $p_1 + p_2 + p_3 = q_1 + q_2 + q_3 = 1$ and 
\begin{align*}
p_1 \vphi(\beta) + p_2 \vphi(\gamma) + p_3 \vphi(\delta) &> q_1 \vphi(\beta) + q_2 \vphi(\gamma) + q_3 \vphi(\delta),\\
p_1 \vphi(a\beta) + p_2 \vphi(a\gamma) + p_3 \vphi(a\delta) &< q_1 \vphi(a\beta) + q_2 \vphi(a\gamma) + q_3 \vphi(a\delta).
\end{align*}
\end{prop}
\begin{prop}
\label{mdp_pm}
If a continuous payoff is positionally determined in MDPs, then this payoff is prefix-monotone.
\end{prop}

We also show that these two necessary conditions imply that $\vphi$ is multi-discounted. 

\begin{prop}
\label{mdp_technical}
Let $A$ be a finite set and $\vphi\colon A^\omega\to\mathbb{R}$ be a continuous prefix-monotone payoff. Assume that there are no $a\in A$, $\beta, \gamma, \delta\in A^\omega$, $(p_1, p_2, p_3), (q_1, q_2, q_3) \in [0, +\infty)^3$ such that $p_1 + p_2 + p_3 = q_1 + q_2 + q_3 = 1$ and 
\begin{align*}
p_1 \vphi(\beta) + p_2 \vphi(\gamma) + p_3 \vphi(\delta) &> q_1 \vphi(\beta) + q_2 \vphi(\gamma) + q_3 \vphi(\delta),\\
p_1 \vphi(a\beta) + p_2 \vphi(a\gamma) + p_3 \vphi(a\delta) &< q_1 \vphi(a\beta) + q_2 \vphi(a\gamma) + q_3 \vphi(a\delta).
\end{align*}
Then $\vphi$ is a multi-discounted payoff.
\end{prop}
Note that Proposition \ref{mdp_pm} is already proved. Indeed, in Section \ref{sec:results} we have shown that for any continuous payoff which is not prefix-monotone, there exists a game graph where $\vphi$ is not positional. This game graph had the following feature: all its nodes were controlled by Max. Thus, this game graph is a deterministic MDP, which means that any continuous payoff which is not prefix monotone is not positionally determined in MDPs.

To finish our proof of Theorem \ref{thm:mdp}, it remains to prove Propositions \ref{dist_nec} and \ref{mdp_technical}.
\subsection{Proof of Proposition \ref{dist_nec}}
Assume for contradiction that such $a, \beta, \gamma, \delta$, $(p_1, p_2, p_3)$ and $(q_1, q_2, q_3)$ exist. By the continuity of $\vphi$, we may assume that $\beta, \gamma$ and $\delta$ are ultimately periodic. We construct an $A$-labeled MDP $\mathcal{M}$ where $\vphi$ has no optimal positional strategy. To define $\mathcal{M}$, consider an $A$-labeled game graph from Figure \ref{fig:st_counter}.

\begin{figure}[h!]
\centering
\begin{tikzpicture}[gamegraph,x=7mm,y=7mm]
   \node[draw, circle] (u) {$u$};
   \node[draw, circle, right = 4 of u] (v) {$v$};
   \node[draw, circle, above right = 3 and 4 of v] (a) {};
   \node[draw, circle, right = 4 of v] (b) {};
   \node[draw, circle, below right = 3 and 4 of v] (c) {};

  \node[right = 2 of b] {};  

\draw[transition] (u) -- (v) node[midway, above] {$c$};
\draw[transition] (v) -- (a) node[near end, above left] {$\alpha$};
\draw[transition] (v) -- (b) node[near end, above] {$\beta$};
\draw[transition] (v) -- (c) node[near end, below left] {$\gamma$};

\path[transition]
  (a) edge [in=330,out=30,out distance=3cm,in distance=3cm] (a);
\path[transition]
  (b) edge [in=330,out=30,out distance=3cm,in distance=3cm]
    (b);
\path[transition]
  (c) edge [in=330,out=30,out distance=3cm,in distance=3cm] (c);

\end{tikzpicture}
  \caption{A graph for an MDP where $\varphi$ has no optimal positional strategy.}
\label{fig:st_counter}
\end{figure}
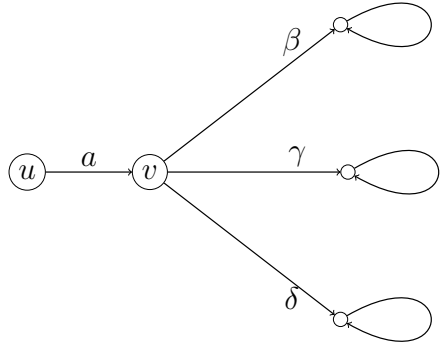

In this graph there are exactly 3 infinite paths (``lassos'') $\mathcal{P}_1$, $\mathcal{P}_2$, $\mathcal{P}_3$ that start at $v$. We label their edges in such a way that $\lab(\mathcal{P}_1) = \beta, \lab(\mathcal{P}_2) = \gamma, \lab(\mathcal{P}_3) = \delta$. This is possible because $\beta, \gamma$ and $\delta$ are ultimately periodic.

Next, we turn this graph into an MDP (formally, nodes of the graph will be states of the MDP).  There will be two actions available at the node $v$. Both will be distributed on the set of successors of $v$. One gives a probability $p_i$ to the successor which leads to the lasso $\mathcal{P}_i$, for $i = 1, 2, 3$. The other gives a probability $q_i$   to the successor which leads to the lasso $\mathcal{P}_i$, for $i = 1, 2, 3$.   For each node different from $v$ there will be only one action with the source in this node, leading with probability $1$ to its unique successor.

It remains to define the labeling function of $\mathcal{M}$. Fix a transition. It is over some edge of the graph from Figure \ref{fig:st_counter}. We define the label this transition as the label of this edge. This concludes the description of $\mathcal{M}$.

To show that $\vphi$ is not positional in $\mathcal{M}$, note that in $\mathcal{M}$ there are exactly $2$ positional strategies, $\sigma_p$ and $\sigma_q$, corresponding to two actions available at $v$. We show that none of these two strategies is optimal.

It is easy to see that:
\begin{align*}
\mathbb{E} \varphi \circ \lab\big(\mathcal{P}_u^{\sigma_p}\big) &= p_1 \cdot \varphi(a\beta) + p_2 \cdot \varphi(a\gamma) + p_3 \cdot \varphi(a\delta),\\
\mathbb{E} \varphi \circ \lab\big(\mathcal{P}_v^{\sigma_p}\big) &= p_1 \cdot \varphi(\beta) + p_2 \cdot \varphi(\gamma) + p_3 \cdot \varphi(\delta),\\
\mathbb{E} \varphi \circ \lab\big(\mathcal{P}_u^{\sigma_q}\big) &= q_1 \cdot \varphi(a\beta) + q_2 \cdot \varphi(a\gamma) + q_3 \cdot \varphi(a\delta),\\
\mathbb{E} \varphi \circ \lab\big(\mathcal{P}_v^{\sigma_q}\big) &= q_1 \cdot \varphi(\beta) + q_2 \cdot \varphi(\gamma) + q_3 \cdot \varphi(\delta).
\end{align*}
Due to our assumptions about $(p_1, p_2, p_3), (q_1, q_2, q_3)$, we obtain:
\[\mathbb{E} \varphi \circ \lab\big(\mathcal{P}_u^{\sigma_p}\big) < \mathbb{E} \varphi \circ \lab\big(\mathcal{P}_u^{\sigma_q}\big), \qquad\mathbb{E} \varphi \circ \lab\big(\mathcal{P}_v^{\sigma_p}\big) >  \mathbb{E} \varphi \circ \lab\big(\mathcal{P}_v^{\sigma_q}\big).\]
Therefore, neither $\sigma_p$ nor $\sigma_q$ is optimal.
\subsection{Proof of Proposition \ref{mdp_technical}}

If $\vphi(\gamma) = \vphi(\delta)$ for all $\gamma, \delta\in A^\omega$, then clearly $\vphi$ is multi-discounted (one can define $\lambda(a) = 0, w(a) = \vphi(\gamma)$ for all $a\in A$ and for an arbitrary $\gamma\in A^\omega$). In what follows, we fix any $\gamma, \delta\in A^\omega$ with $\vphi(\gamma) \neq\vphi(\delta)$.
First we derive from the conditions of Proposition \ref{mdp_technical} the following:
\begin{lem}
\label{next_lemma}
 For all $a\in A$ there exist $\lambda(a), w(a) \in \mathbb{R}$ such that for any $\beta \in A^\omega$ we have:
\[\varphi(a\beta) = \lambda(a) \varphi(\beta) + w(a).\]
\end{lem}
\begin{proof}
The following system in $(\lambda, w)$ has a unique solution:
\begin{equation}
\label{tag}
\begin{pmatrix} \vphi(a\gamma) \\ \vphi(a\delta) \end{pmatrix} = \begin{pmatrix} \vphi(\gamma) & 1 \\ \vphi(\delta) & 1 \end{pmatrix} \cdot \begin{pmatrix} \lambda \\ w \end{pmatrix},
\end{equation}
(because $\vphi(\gamma) \neq \vphi(\delta)$). Let its solution be $(\lambda(a), w(a))$. We show that $\vphi(a\beta) = \lambda(a) \vphi(\beta) + w(a)$ for all $\beta\in A^\omega$.
Let us first show that
\begin{equation}
\label{some_det}
\det \begin{pmatrix} 1 & 1 & 1 \\ \varphi(\beta) & \varphi(\gamma) & \varphi(\delta) \\ \varphi(a\beta) & \varphi(a\gamma) & \varphi(a\delta) \end{pmatrix} = 0.
\end{equation}
Indeed, otherwise there exists a vector $(x, y, z) \in \mathbb{R}^3$ such that 
\begin{equation}
\label{det_solution}
 \begin{pmatrix} 1 & 1 & 1 \\ \varphi(\beta) & \varphi(\gamma) & \varphi(\delta) \\ \varphi(a\beta) & \varphi(a\gamma) & \varphi(a\delta) \end{pmatrix} \cdot \begin{pmatrix} x \\ y \\ z \end{pmatrix} =  \begin{pmatrix} 0 \\ 1 \\ -1 \end{pmatrix}.
\end{equation}
 Let $P_1, P_2, P_3, Q_1, Q_2, Q_3$ be any positive real numbers such that $x = P_1 - Q_1, y = P_2 - Q_2, z = P_3 - Q_3$. From the first equality in \eqref{det_solution} it follows that  $P_1 + P_2 + P_3 = Q_1 + Q_2 + Q_3 = S > 0$.
Define $p_i = P_i/S, q_i = Q_i/S$ for  $i\in\{1, 2, 3\}$. Observe that $a, \beta, \gamma, \rho, (p_1, p_2, p_3), (q_1, q_2, q_3)$ violate the conditions of Proposition \ref{mdp_technical} (this can be seen from the second and the third equalities in \eqref{det_solution}), contradiction. Therefore \eqref{some_det} is proved.

The first two rows of the matrix from \eqref{some_det} are linearly independent because $\vphi(\gamma) \neq \vphi(\delta)$. Hence, the third one must be a linear combination of the first two. I.e., there must exist $\lambda, w\in \mathbb{R}$ such that 
\[(\varphi(a\beta), \varphi(a\gamma),  \varphi(a\delta)) = \lambda (\varphi(\beta), \varphi(\gamma), \varphi(\delta)) + w (1, 1, 1).\]
From the second and the third coordinate we conclude that $(\lambda, w)$ must be a solution to \eqref{tag}, so $\lambda = \lambda(a), w = w(a)$. Now, looking at the first coordinate, we obtain that $\vphi(a\beta) = \lambda(a)\vphi(\beta) + w(a)$, as required.
\end{proof}

From now on let $\lambda(a), w(a)$ for $a\in A$ be as in Lemma \ref{next_lemma}. Let us show that $\lambda(a)\in [0, 1)$ for all $a\in A$.

 Assume first that for some $a\in A$ we have $\lambda(a) < 0$. Without loss of generality, we may also assume that $\varphi(\gamma) < \varphi(\delta)$. Then $\varphi(a\gamma) = \lambda(a) \varphi(\gamma) + w(a)>\lambda(a) \varphi(\delta) + w(a)  = \varphi(a\delta)$. But $\vphi$ is prefix-monotone, so this is impossible.

Next, assume for contradiction that $\lambda(a) \ge 1$ for some $a\in A$. Consider the following two sequences $\{x_n\}_{n\in\mathbb{N}}$ and $\{y_n\}_{n\in\mathbb{N}}$ of real numbers: \[x_n = \varphi(\underbrace{aa\ldots a}_n \gamma), \qquad y_n = \varphi(\underbrace{aa\ldots a}_n \delta).\]
 Note that by our choice of $\gamma$ and $\delta$, we have $x_0 = \vphi(\gamma) \neq \vphi(\delta) = y_0$. Next, since $\varphi$ is continuous, we have:
\begin{equation}
\label{two_limits}
\lim\limits_{n \to\infty} x_n = \lim\limits_{n \to\infty} y_n = \varphi(aaa\ldots).
\end{equation}
On the other hand, we can compute $x_n$ and $y_n$ through Lemma \ref{next_lemma}:
\begin{align}
\label{directly_compute_1}
x_n &= \lambda(a)^n x_0 + w(a) (1 + \lambda(a) + \ldots + \lambda(a)^{n -1}),\\
\label{directly_compute}
 y_n &= \lambda(a)^n y_0 + w(a) (1 + \lambda(a) + \ldots + \lambda(a)^{n -1}).
\end{align}
First, consider the case $\lambda(a) = 1$. Then $x_n$ and $y_n$ look as follows:
\[x_n = x_0 + n w(a), \qquad y_n = y_0 + n w(a).\]
If $w(a) \neq 0$, then the sequences $\{x_n\}_{n\in\mathbb{N}}$ and $\{y_n\}_{n\in\mathbb{N}}$ are not convergent, contradiction with \eqref{two_limits}. If $w(a) = 0$, then one sequence converges to $x_0$, and the other to $y_0$. But $x_0 \neq y_0$, so this again gives a contradiction with \eqref{two_limits}.

Now, consider the case $\lambda(a) > 1$. Then we can rewrite (\ref{directly_compute_1}--\ref{directly_compute}) as follows:
\[x_n = \lambda(a)^n\left(x_0 + \frac{w(a)}{\lambda(a) - 1}\right) -  \frac{w(a)}{\lambda(a) - 1}, \qquad y_n = \lambda(a)^n\left(y_0 + \frac{w(a)}{\lambda(a) - 1}\right) -  \frac{w(a)}{\lambda(a) - 1}.\]
Since $x_0 \neq y_0$, the coefficient before $\lambda(a)^n$ is non-zero in at least one of these expressions. Hence either $\{x_n\}_{n\in\mathbb{N}}$ or $\{y_n\}_{n\in\mathbb{N}}$ diverge. This contradicts \eqref{two_limits}.

\bigskip

We have established that $\lambda(a) \in [0, 1)$ for every $a\in A$. All that remains to do is to show that $\varphi$ satisfies \eqref{multi_disc}. 
For that, we again employ the continuity of $\varphi$. Take any $\beta\in A^\omega$. Note that by Lemma \ref{next_lemma} we have:
\begin{equation}
\label{sorry}
\varphi(a_1 a_2 \ldots a_n \beta) = \lambda(a_1) \cdot \ldots \cdot  \lambda(a_{n}) \vphi(\beta) + \sum\limits_{i = 1}^n \lambda(a_1) \cdot \ldots \lambda(a_{i - 1}) \cdot w(a_i).
\end{equation}

We know that $\lambda(a_i)$ are all from $[0, 1)$. Since the set $A$ is finite, all $\lambda(a_i)$ are bounded from above by some number smaller than 1. Hence, the first term in the right-hand side of \eqref{sorry} converges to $0$ as $n\to\infty$. On the other hand, the second term in the right-hand side of \eqref{sorry} converges to the series from the right-hand side of \eqref{multi_disc}. Finally, due to the continuity of $\varphi$, the left-hand side of \eqref{sorry} converges to $\vphi(a_1 a_2 a_3\ldots)$. Thus, $\vphi$ is multi-discounted.

\section{Discussion}
\label{sec:disc}

As Gimbert and Zielonka show by their characterization of the class of positionally determined payoffs~\cite{gimbert2005games}, positional determinacy can always be proved by an inductive argument. Does the same hold for two other techniques that we have considered in the paper -- the fixed point technique and the strategy improvement technique?  Is it at least true for prefix-independent positionally determined payoffs? E.g., for the mean payoff, a major example of a prefix-independent positionally determined payoff, both the strategy improvement and the fixed point arguments are applicable~\cite{gurvich1988cyclic,kohlberg1980invariant}. 

These questions are specifically interesting for the strategy improvement argument. A fact that a non-optimal positional strategy can be improved by modifying it in a single node lies in the core of all subexponential-time algorithms for positionally determined payoffs~\cite{halman2007simple,bjorklund2005combinatorial}. So if we could always prove positional determinacy by strategy improvement, maybe we can also solve any positionally determined payoff in subexponential time?

Finally, it would be interesting to obtain an explicit description of other classes of positionally determined payoffs -- for example, of prefix-independent positionally determined payoffs.

\bibliographystyle{alphaurl}
\bibliography{ref}
\end{document}